\documentclass[11pt,a4paper]{article}

\usepackage{vmargin}
\setmarginsrb{1in}{1in}{1in}{1in}{0cm}{0cm}{8mm}{8mm}
\usepackage[margins=normal,indent=normal,title=normal,leading=normal,bibliography=normal,bibnotes=normal,charwidths=normal,tracking=normal,wordspacing=normal,paragraphs=tight,sections=normal,lists=tight,floats=tight]{savetrees}

\usepackage{url}
\usepackage{graphicx}
\usepackage{xspace}
\newcommand{\ig}[2][scale=1.0]{\includegraphics[#1]{#2}}
\usepackage{hyperref}
\usepackage{booktabs}
\usepackage{color}
\usepackage{tikz}
\usetikzlibrary{calc,3d,decorations.markings}
\usepackage{caption}
\usepackage{subcaption}
\usepackage{algorithmicx}
\usepackage{algpseudocode}
\usepackage{enumitem}
\usepackage{amsmath,amsthm,amssymb}
\usepackage{calc}

\newcommand{\etal}{et al.~}
\newcommand{\etaln}{et al.}
\newcommand{\ie}{i.e.~}

\newcommand{\mc}{\mathcal}
\newcommand{\poly}{\mathrm{poly}}

\newcommand{\tw}{\mathsf{tw}}
\newcommand{\cI}{\mathcal{I}}

\newcommand{\type}[0]{\textup{type}}
\newcommand{\DP}[0]{\textup{DP}}

\DeclareSymbolFont{AMSb}{U}{msb}{m}{n}
\DeclareSymbolFontAlphabet{\mathbb}{AMSb}

\newtheorem{theorem}{Theorem}[section]

\newtheorem{lemma}[theorem]{Lemma}

\newtheorem{definition}[theorem]{Definition}

\newcommand{\knip}[1]{}

\newcommand{\problemSC}{\textsc{Steiner Multicut}}
\newcommand{\problemESC}{\textsc{Edge Steiner Multicut}}
\newcommand{\problemNSC}{\textsc{Node Steiner Multicut}}
\newcommand{\problemMWC}{\textsc{Multiway Cut}}

\newcommand{\problemMC}{\textsc{Multicut}}
\newcommand{\problemMCC}{\textsc{Multicolored Clique}}
\newcommand{\problemEMC}{\textsc{Edge Multicut}}
\newcommand{\problemNMC}{\textsc{Node Multicut}}
\newcommand{\problemRNMC}{\textsc{Restr.\ Node Multicut}}
\newcommand{\problemRNSC}{\textsc{Restr.\ Node Steiner Multicut}}
\newcommand{\problemHS}{\textsc{Hitting Set}}
\newcommand{\problemSetC}{\textsc{Set Cover}}
\newcommand{\problemUSC}{\textsc{Unipedal Steiner Multicut}}

\newcommand{\problemMPSC}{\textsc{Multipedal Steiner Multicut}}
\newcommand{\problemNAE}{\textsc{NAE-Integer-3-SAT}}
\newcommand{\NAE}{\textup{NAE}}

\newcommand{\nameESC}{edge Steiner multicut}
\newcommand{\nameRNSC}{restr.\ node Steiner multicut}
\newcommand{\nameNSC}{node Steiner multicut}

\newcommand{\FPT}{\ensuremath{\mathsf{FPT}}}

\newcommand{\tG}{\tilde{G}}
\newcommand{\tT}{\tilde{T}}
\newcommand{\hG}{\hat{G}}
\newcommand{\hT}{\hat{T}}

\newcommand{\defparproblem}[4]{

\bigskip
\noindent\fbox{
	\begin{minipage}{.96\linewidth}
	\textsc{#1}
	
	\smallskip
	\noindent\begin{tabular}{@{}l@{ }l}
	\emph{Input:} & \begin{minipage}[t]{\linewidth-\widthof{Parameter:\ \ \ \ }} #3\end{minipage}\\[.3pt]
	\emph{Parameter:} & \begin{minipage}[t]{\linewidth-\widthof{Parameter:\ \ \ \ }} #2\end{minipage}\\[.3pt]
	\emph{Question:} & \begin{minipage}[t]{\linewidth-\widthof{Parameter:\ \ \ \ }} #4\end{minipage}
	\end{tabular}
	\end{minipage}
}

\medskip
}
\newcommand{\defproblem}[3]{

\bigskip
\noindent\fbox{
	\begin{minipage}{.96\linewidth}
	\textsc{#1}
	
	\smallskip
	\noindent\begin{tabular}{@{}l@{ }l}
	\emph{Input:} & \begin{minipage}[t]{\linewidth-\widthof{Parameter:\ \ \ \ }} #2\end{minipage}\\
	\emph{Question:} & \begin{minipage}[t]{\linewidth-\widthof{Parameter:\ \ \ \ }} #3\end{minipage}
	\end{tabular}
	\end{minipage}
}

\medskip
}

\begin{document}

\title{Parameterized Complexity Dichotomy for {\sc Steiner Multicut}\thanks{An extended abstract of this paper appeared in Mayr, E.W., Ollinger, N.\ (eds.) \emph{32nd International Symposium on Theoretical Aspects of Computer Science (STACS 2015)}, LIPIcs volume 30, Schloss Dagstuhl, 2015, pp.~157--170.}}
\author{%
Karl Bringmann\thanks{Institute of Theoretical Computer Science, ETH Zurich, Switzerland. \texttt{karlb@inf.ethz.ch}. Supported by the ETH Zurich Postdoctoral Fellowship Program.}
\and
Danny Hermelin\thanks{Ben Gurion University of the Negev, Israel. \texttt{hermelin@bgu.ac.il}. Supported by the People Programme (Marie Curie Actions) of the European Union's Seventh Framework Programme (FP7/2007-2013) under REA grant agreement number 631163.11, and by the Israel Science Foundation (grant no.\ 551145/14).}
\and
Matthias Mnich\thanks{Universit\"at Bonn, Germany. \texttt{mmnich@uni-bonn.de}. Supported by ERC Starting Grant~306465.}
\and
Erik Jan van Leeuwen\thanks{Max-Planck Institut f\"{u}r Informatik, Saarbr\"{u}cken, Germany. \texttt{erikjan@mpi-inf.mpg.de}.}
}

\maketitle

\begin{abstract}
  We consider the \problemSC{} problem, which asks, given an undirected graph $G$, a collection $\mc{T} = \{T_{1},\ldots,T_{t}\}$, $T_{i} \subseteq V(G)$, of terminal sets of size at most $p$, and an integer $k$, whether there is a set $S$ of at most $k$ edges or nodes such that of each set $T_{i}$ at least one pair of terminals is in different connected components of $G \setminus S$.
  This problem generalizes several well-studied graph cut problems, in particular the \problemMC{} problem, which corresponds to the case $p = 2$.
  The \problemMC{} problem was recently shown to be fixed-parameter tractable for the parameter $k$ [Marx and Razgon, Bousquet~\etaln, STOC 2011].
  The question whether this result generalizes to \problemSC{} motivates the present work.

  We answer the question that motivated this work, and in fact provide a dichotomy of the parameterized complexity of \problemSC{} on general graphs.
  That is, for any combination of $k$, $t$, $p$, and the treewidth $\tw(G)$ as constant, parameter, or unbounded, and for all versions of the problem (edge deletion and node deletion with and without deletable terminals), we prove either that the problem is fixed-parameter tractable or that the problem is hard ($\mathsf{W}[1]$-hard or even (para-)$\mathsf{NP}$-complete).
  Among the many results in the paper, we highlight that:
  \begin{itemize}
    \item The edge deletion version of \problemSC{} is fixed-parameter tractable for the parameter $k+t$ on general graphs (but has no polynomial kernel, even on trees).
    We present two independent proofs of fixed-parameter tractability. The first proof uses the randomized contractions technique of Chitnis~et al. The second proof relies on several new structural lemmas, which decompose the Steiner cut into important separators and minimal $s$-$t$ cuts.
    \item In contrast, both node deletion versions of \problemSC{} are $\mathsf{W}[1]$-hard for the parameter $k+t$ on general graphs.
    \item All versions of \problemSC{} are $\mathsf{W}[1]$-hard for the parameter $k$, even when $p=3$ and the graph is a tree plus one node.
    This means that the mentioned results of Marx and Razgon, and Bousquet~\etal do not generalize to even the most basic instances of \problemSC{}.
  \end{itemize}
  Since we allow $k$, $t$, $p$, and $\tw(G)$ to be any constants, our characterization includes a dichotomy for \problemSC{} on trees (for $\tw(G) = 1$) as well as a polynomial time versus $\mathsf{NP}$-hardness dichotomy (by restricting $k,t,p,\tw(G)$ to constant or unbounded).
\end{abstract}

\section{Introduction}
\label{sec:introduction}
Graph cut problems are among the most fundamental problems in algorithmic research.
The classic result in this area is the polynomial-time algorithm for the $s$--$t$ cut problem of Ford and Fulkerson~\cite{FordFulkerson1956} (independently proven by Elias et al.~\cite{EliasEtAl1956} and Dantzig and Fulkerson~\cite{DantzigFulkerson1956}).
This result inspired a research program to discover the computational complexity of this problem and of more general graph cut problems.
One well-studied generalization of the $s$--$t$ cut problem is the \problemMC{} problem, in which we want to disconnect $t$ pairs of nodes instead of just one pair.
In a recent major advance of the research program on graph cut problems, Bousquet~\etal\cite{BousquetEtAl2011} and Marx and Razgon~\cite{MarxRazgon2011} showed that \problemMC{} is fixed-parameter tractable in the size $k$ of the cut only, meaning that it has an algorithm running in time $f(k)\cdot \poly(|V(G)|)$ for some function $f$, resolving a longstanding problem in parameterized complexity (with many papers~\cite{Marx2006,Xiao2010,Guillemot2011,MarxEtAl2013} building up to this result).

In this paper, we continue the research program on generalized graph cut problems, and consider the \problemSC{} problem.
This problem was proposed by Klein~\etal\cite{KleinEtAl1997}, and appears in several versions, depending on whether we want to delete edges or nodes, and whether we are allowed to delete terminal nodes.
Formally, these versions of the \problemSC{} problem are defined as follows:

\defproblem{\{Edge, Node, Restricted Node\} Steiner Multicut}{An undirected graph $G$ with terminal sets $T_1,\hdots,T_t\subseteq V(G)$, and an integer $k\in\mathbb N$.}{Find a set $S$ of $k$ $\{$edges, nodes, non-terminal nodes$\}$ such that for $i=1,\ldots,t$ and at least one pair $u,v\in T_i$ there is no $u-v$ path in $G \setminus S$.}

Observe that \problemMC{} is the special case of \problemSC{} in which each terminal set has size two.
In general, the terminal sets of \problemSC{} can have arbitrary size, and we use $p$ to denote $\max_i |T_i|$.  

The complexity of \problemSC{} has been investigated extensively, but so far only from the perspective of approximability.
This line of work was initiated by Klein~\etal\cite{KleinEtAl1997}, who gave an LP-based $O(\log^3(kp))$-approximation algorithm.
The approximability of several variations of the problem has also been considered~\cite{YuCheriyan1995,GargEtAl1996,AvidorLangberg2007}; in particular, Garg~\etal\cite{GargEtAl1996} give an $O(\log t)$-approximation algorithm for \problemMC{}.
On the hardness side, even \problemMC{} is $\mathsf{APX}$-hard~\cite{DahlhausEtAl1994,CalinescuEtAl2003} and cannot be approximated within any constant factor assuming the Unique Games Conjecture~\cite{ChawlaEtAl2006}.
We also remark that Steiner cuts (the case when $t=1$) are of interest: they are an ingredient in several LP-based approximation algorithms (for example for \textsc{Steiner Forest}~\cite{AgrawalEtAl1995,KonemannEtAl2008}) and there is a connection to the number of edge-disjoint Steiner trees that each connect all terminals~\cite{Lau2007}.
To the best of our knowledge, however, \problemSC{} in its general form has not yet been considered from the perspective of parameterized complexity.

\subsection{Our Contribution}
\label{sec:ourcontributions}%
In this paper, we fully chart the (parameterized) complexity landscape of \problemSC{} according to $k$, $t$, $p$ (defined as above), and the treewidth $\tw(G)$. For all three versions of \problemSC{}, for each possible combination of $k$, $t$, $p$, and $\tw(G)$, where each may be either chosen as a constant, a parameter, or unbounded, we consider the complexity of \problemSC{}. We show a complete dichotomy: either we provide a fixed-parameter algorithm with respect to the chosen parameters, or we prove a $\mathsf{W}[1]$-hardness or (para-)$\mathsf{NP}$-completeness result that rules out a fixed-parameter algorithm (unless many canonical $\mathsf{NP}$-complete problems have subexponential- or polynomial-time algorithms respectively).

The dichotomy is composed of three main results, along with many smaller ones (see Table~\ref{tab:undirected_parameterized_complexity_results}). These three main results are stated in the three theorems below:
  
\begin{table}[t]
  \centering{\footnotesize
  \begin{tabular}{ll|c|c|c}
    \toprule
                           &            & {\sc Edge}                                 & {\sc Node}       & {\sc Restr.\ Node}\\
      constants            & params & {\sc Steiner MC}                     & {\sc Steiner MC}        & {\sc Steiner MC}\\
      
    \midrule
       
      $k$ & --- & poly (Sect.~\ref{sec:easyandknown}) & poly (Sect.~\ref{sec:easyandknown}) & poly (Sect.~\ref{sec:easyandknown})\\
      
      $t \leq 2$ & --- & poly (Sect.~\ref{sec:easyandknown}) & & poly (Sect.~\ref{sec:easyandknown})\\
      
      $t = 3, p = 2$ & --- & $\mathsf{NP}$-h~\cite{DahlhausEtAl1994} &  & $\mathsf{NP}$-h~\cite{DahlhausEtAl1994} \\
      
      --- & $k,t$ & ${}^{\star}\mathsf{FPT}$ (Thm.~\ref{thm:intro:k+t-edge}) & ${}^{\star}\mathsf{W}[1]$-h (Thm.~\ref{thm:intro:k+t}) & ${}^{\star}\mathsf{W}[1]$-h (Thm.~\ref{thm:intro:k+t}) \\
      
      ---                  & $k,p,t$    & & $\mathsf{FPT}$ (Sect.~\ref{sec:easyandknown}) & $\mathsf{FPT}$ (Sect.~\ref{sec:easyandknown}) \\

      $t$ & $k$ &  &  & $\mathsf{FPT}$ (Sect.~\ref{sec:easyandknown})\\
      
      $p = 2$              & $k$        & $\mathsf{FPT}$~\cite{BousquetEtAl2011,MarxRazgon2011} & $\mathsf{FPT}$~\cite{BousquetEtAl2011,MarxRazgon2011} & $\mathsf{FPT}$~\cite{BousquetEtAl2011,MarxRazgon2011}\\
      
      $p = 3, \tw = 2$              & $k$        & ${}^{\star}\mathsf{W}[1]$-h (Thm.~\ref{thm:intro:k}) & ${}^{\star}\mathsf{W}[1]$-h (Thm.~\ref{thm:intro:k}) & ${}^{\star}\mathsf{W}[1]$-h (Thm.~\ref{thm:intro:k})\\
      
       --- & $t,\tw$ & ${}^{\star}\mathsf{FPT}$ (Thm.~\ref{thm:t+tw}) & ${}^{\star}\mathsf{FPT}$ (Thm.~\ref{thm:t+tw}) & ${}^{\star}\mathsf{FPT}$ (Thm.~\ref{thm:t+tw})\\
      
      
      $\tw=1$ & --- &  & ${}^{\star}$poly (Thm.~\ref{thm:node:easy}) & \\
      
      $\tw = 1$ & $k$ & ${}^{\star}\mathsf{W}[2]$-h (Thm.~\ref{thm:edge:hard-k}) & & ${}^{\star}\mathsf{W}[2]$-h (Thm.~\ref{thm:rnode:hard-k}) \\
      
      $\tw=1$ & $k,p$ & ${}^{\star}$$\mathsf{FPT}$ (Thm.~\ref{thm:edge:easy-kp}) & & ${}^{\star}$$\mathsf{FPT}$ (Thm.~\ref{thm:rnode:easy-kp}) \\
      
      $\tw=1,p= 2$ & --- & $\mathsf{NP}$-h~\cite{CalinescuEtAl2003} & & $\mathsf{NP}$-h~\cite{CalinescuEtAl2003} \\
      
      $\tw=2, p=2$ & --- & & $\mathsf{NP}$-h~\cite{CalinescuEtAl2003} &  \\
      
    \bottomrule
  \end{tabular}}
  
  \centering
  \caption{Summary of known and new complexity results for \problemSC{}, where new results are marked with ${}^{\star}$; the other entries are either known or follow easily from known results in the literature. Only maximal $\mathsf{FPT}$ results and minimal $\mathsf{W}[\cdot]$- or $\mathsf{NP}$-hardness results are listed; empty cells are dominated by other results. E.g.\ \problemESC{} with parameter $t$ is hard, since it is already $\mathsf{NP}$-hard for $t=3, p=2$. For \problemNSC{}, one also has to apply the rule that $k < t$ (see Section~\ref{sec:easyandknown}) to generate a full characterization of all cases. Tree diagrams of this table are offered in Appendix~\ref{sec:trees}.}
\label{tab:undirected_parameterized_complexity_results}
\end{table}

\begin{theorem} \label{thm:intro:k+t-edge}%
\problemESC{} is fixed-parameter tractable for the parameter $k+t$.
\end{theorem}

\begin{theorem} \label{thm:intro:k+t}%
\problemNSC{} and \problemRNSC{} are $\mathsf{W}[1]$-hard for the parameter $k+t$.
\end{theorem}

\begin{theorem} \label{thm:intro:k}
\problemNSC{}, \problemESC{}, and \problemRNSC{} are $\mathsf{W}[1]$-hard for the parameter $k$, even if $p=3$ and $\tw(G)=2$.
\end{theorem}

\noindent Observe the sharp gap described by Theorem~\ref{thm:intro:k+t-edge} and Theorem~\ref{thm:intro:k+t} between the parameterized complexity of the edge deletion version versus the node deletion version; this gap does not \mbox{exist} for the {\sc Multicut} problem. We also note that Theorem~\ref{thm:intro:k} implies that the fixed-parameter algorithms for \problemMC{} for parameter $k$~\cite{BousquetEtAl2011,MarxRazgon2011} do not generalize to \problemSC{}.

To obtain the fixed-parameter algorithm of Theorem~\ref{thm:intro:k+t-edge}, we have to avoid the brute-force choice of a pair of separated terminals of each terminal set: Although one can trivially reduce every instance of the {\sc Edge Steiner Multicut} problem to at most ${p \choose 2}^{t}$ instances of \problemMC{} parameterized by $k$, this only yields an $f(k) \cdot n^{O(t)}$-time algorithm (for unbounded~$p$).
Our contribution in Theorem~\ref{thm:intro:k+t-edge} is that we improve on this simple algorithm and obtain a runtime of $f(k,t) \cdot n^{O(1)}$. We give two independent proofs of Theorem~\ref{thm:intro:k+t-edge}:
\begin{itemize}
\item Our first proof uses the recent technique of Chitnis~\etal\cite{ChitnisEtAl2012} known as randomized contractions (even though the technique actually yields deterministic algorithms). The rough idea of the algorithm is to first determine a large subgraph $G'$ of the input graph $G$, such that $G'$ has no small cut and only has a small interface (\ie a small number of vertices that connect the subgraph to the rest of the graph). We can then branch on the behavior of a solution on the interface to determine a set $U \subseteq E(G')$ of `useless' edges, in the sense that when $U$ is contracted in $G$ a smallest solution (of size at most $k$) persists in the remaining graph. By iterating this procedure, we can reduce the size of the graph until it is small enough to be handled by exhaustive enumeration.
Our algorithm is similar to the one for \textsc{Edge Multiway Cut-Uncut} in the paper by Chitnis~\etal\cite{ChitnisEtAl2012}; however, in contrast to that problem, there seems to be no straightforward projection of the instance onto $G'$ in our case, and therefore more involved arguments are needed to determine the set $U$.
\item Our second proof is based on several novel structural lemmas that show that a minimal edge Steiner cut can be decomposed into important separators and minimal $s$-$t$ cuts.
Using a branching strategy, we ascertain the topology of the decomposition that is promised by the structural lemmas.
Since there are only few important separators of bounded size~\cite{Marx2006,ChenEtAl2009,MarxRazgon2011} and all relevant minimal $s$-$t$ cuts lie in a graph of bounded treewidth (following the treewidth reduction techniques of Marx~\etal\cite{MarxEtAl2013}), we can then optimize over important separators and minimal $s$-$t$ cuts.
\end{itemize}
The advantage of the first algorithm over the second is that it runs in single-exponential time, instead of double-exponential time. However, the second algorithm is slightly faster in terms of $n$. Moreover, as part of the correctness proof of the second algorithm, we present some structural lemmas that give additional insight into the properties of the cut, which may be of independent interest. Therefore, we present both algorithms.

The $\mathsf{W}[1]$-hardness results of Theorem~\ref{thm:intro:k+t} and~\ref{thm:intro:k} all rely on reductions from the \problemMCC{} problem~\cite{FellowsEtAl2009}. 
For the proof of Theorem~\ref{thm:intro:k}, we introduce a novel intermediate problem, \problemNAE{}, which is an integer variant of the better known \textsc{Not-All-Equal-3-SAT} problem.
We show that \problemNAE{} is $\mathsf{W}[1]$-hard parameterized by the number of variables.
This is a powerful starting point for parameterized hardness reductions and should turn out to be useful to prove the hardness of other problems.

To complete our dichotomy, the second part of our paper charts the full (parameterized) complexity of \problemSC{} on trees, that is, for graphs $G$ with $\tw(G)=1$.
In fact, some of the hardness results that we prove for \problemSC{} on general graphs even hold for trees.
We also show that many of the results for trees do not carry over to graphs of bounded treewidth, the only exception being a fixed-parameter algorithm for parameters $\tw(G) + t$.

We remark that our characterization induces a polynomial time vs. $\mathsf{NP}$-hardness dichotomy for \problemSC{}, i.e., for any choice of $k,p,t,\tw(G)$ as any constants or unbounded (and all three problem variants), we either prove that \problemSC{} is in $\mathsf{P}$ or that it is $\mathsf{NP}$-hard. This characterization can be obtained from Table~\ref{tab:undirected_parameterized_complexity_results} by considering all its polynomial time and $\mathsf{NP}$-hardness results as well as using the rule that any fixed-parameter algorithm induces a polynomial-time algorithm by setting all parameters to $O(1)$.

\subsection{Related Work}
We briefly cite the most relevant results on the parameterized complexity of graph cut problems.
We already mentioned several results on the special case of \problemSC{} when $p=2$ (\problemMC{})~\cite{Marx2006,Xiao2010,Guillemot2011,BousquetEtAl2011,MarxRazgon2011,MarxEtAl2013}.
\problemMC{} is itself a generalization of \problemMWC{}, also known as \textsc{Multiterminal Cut}, where the goal is to delete $k$ edges or nodes to separate all terminals from each other.
This problem is $\mathsf{NP}$-complete even for three terminals~\cite{DahlhausEtAl1994} and has been extensively studied from a parameterized point of view (see, e.g., the work of Cao et al.~\cite{CaoEtAl2013} or Cygan et al.~\cite{CyganEtAl2013}).
The parameterized complexity of many different other graph cut problems has also been considered in recent years~\cite{DowneyEtAl2003,ChitnisEtAl2012,KawarabayashiThorup2011,Marx2006}.
On trees, we only mention here that \problemEMC{} and \problemRNMC{} remain $\mathsf{NP}$-hard~\cite{CalinescuEtAl2003}, but are fixed-parameter tractable~\cite{GuoNiedermeier2005,GuoEtAl2006}.
In contrast, \problemNMC{} has a polynomial-time algorithm on trees~\cite{CalinescuEtAl2003}.

\subsection{Organization}
We begin our exposition in Section~\ref{sec:easyandknown} by giving easy results for certain parameter combinations of \problemSC{}.
Thereafter, we present our fixed-parameter algorithms for \problemESC{} (Theorem~\ref{thm:intro:k+t-edge}) in Section~\ref{sec:ktedge2} and~\ref{sec:ktedge}.
Following this, in the two subsequent sections, we present our $\mathsf{W}[1]$-hardness proofs: the proof of Theorem~\ref{thm:intro:k} in Section~\ref{sec:steinercutsintreewidth}, and of Theorem~\ref{thm:intro:k+t} in Section~\ref{sec:kt}.
Section~\ref{sec:steinercutsintrees} then focuses on trees to complete our dichotomy.
We conclude with some discussion and open problems in Section~\ref{sec:discussion}.
For basic notions of parameterized complextiy as well as the notion of treewidth we refer the reader to Appendix~\ref{sec:basic}.

\section{Easy and Known Results}
\label{sec:easyandknown}
In this section, we collect easy and known results about the {\sc Steiner Multicut} problem. Some of these results are scattered throughout the literature, while others are new.
First, observe that whenever the cut size $k$ is constant, we can solve the problem in polynomial time by simply guessing the desired set $S$ of at most $k$ edges or nodes. 

Furthermore, \problemNSC{} is trivially solvable when $t \leq k$, as in this case we may simply delete an arbitrary terminal node from each set $T_i$, resulting in a solution of size at most $k$; thus, any instance is always a ``yes''-instance in this case. 

We may reduce \problemSC{} to ${p \choose 2}^t$ instances of \problemMC{} by branching for each terminal set over its separated terminals.
Since \problemMC{} is in \FPT\ for parameter $k$, we obtain a fixed-parameter algorithm for \problemSC{} for parameter $k+t+p$.
Also, since ${p \choose 2}^t \le n^{O(t)}$, \problemSC{} is in \FPT\ for parameter $k$ and any constant $t$.

Now, \problemMC{} on instances with $t=1$ (\ie instances that have only one terminal pair $|T_1|=\{s,t\}$) is polynomial-time solvable by running an $s-t$ cut algorithm.
For $t=2$ a result by Yannakakis et al.~\cite[Lemma 1]{YannakakisEtAl1983} also yields a polynomial time algorithm for \problemMC{}.
Again by branching over the separated terminals in both terminal sets, we obtain a polynomial time algorithm for \problemSC{} for $t \le 2$.

When there are three or more terminal sets, then \problemSC{} generalizes \problemMWC{} and thus is $\mathsf{NP}$-complete~\cite{DahlhausEtAl1994} even when $p=2$.

We next show that \problemRNSC{} is as least as hard as \problemNSC{}.
Therefore, whenever \problemNSC{} is $\mathsf{W}[1]$-hard (or $\mathsf{NP}$-hard) for a certain combination of parameters, then so is \problemRNSC{}.
\begin{lemma} \label{lem:NSCtoRNSC}
  Any instance of \problemNSC{} can be reduced in polynomial time to an instance of \problemRNSC{} with the same parameter values $k$, $t$, $p$, and $\mathsf{tw}(G)$. 
\end{lemma}
\begin{proof}
  Take an instance $(G, \mc T, k)$, $\mc T = \{T_1,\ldots,T_t\}$, of \problemNSC{} and transform it to an instance of \problemRNSC{} by adding for each terminal node $v\in T_1\cup\hdots\cup T_t$ a new pendant node $v'$.
  Then replace $v$ by $v'$ in every terminal set $T_i$.
  It is easy to see that the original instance admits a node cut of size $k$ if and only if the new instance admits a node cut of size $k$ that does not use any terminal nodes.
\end{proof}

\section{Tractability for Edge Deletion and Parameter \texorpdfstring{$k+t$}{k + t}: First Proof}
\label{sec:ktedge2}
In this section, we prove Theorem~\ref{thm:intro:k+t-edge}, namely that \problemESC{} parameterized by $k+t$ is fixed-parameter tractable. Here, we give the first proof, which uses the technique of randomized contractions pioneered by Chitnis~\etal\cite{ChitnisEtAl2012}; the second proof is in the next section.
Later we will see that this result is ``maximal'', in the sense that \problemESC{} is $\mathsf{W}[1]$-hard parameterized by $k$ or $t$ alone (this follows from Theorem~\ref{thm:edge:hard-k} and the fact that even \problemEMC{} is $\mathsf{NP}$-hard when $t=3$~\cite{DahlhausEtAl1994} respectively), that the corresponding node deletion problem is $\mathsf{W}[1]$-hard parameterized by $k+t$ (Theorem~\ref{thm:k+t}), and that there exists no polynomial kernel for \problemESC{} parameterized by $k+t$ (Theorem~\ref{thm:edge:nokernel-kt}).

We first state some notions and supporting lemmas from the paper of Chitnis~\etal\cite{ChitnisEtAl2012}, which are needed to make our proof work.
The \emph{identification} of two vertices $v,w \in V(G)$ results in a graph $G'$ by removing $v,w$, adding a new vertex $vw$, and if $v$ or $w$ is an endpoint of an edge, then we replace this endpoint by $vw$. Note that the identification of two vertices does not remove any edges, and generally results in a multigraph (with parallel edges and self-loops). Without confusion, we may sometimes refer to $vw$ by its old names $v$ or $w$.

The \emph{contraction} of an edge $(v,w) \in E(G)$ results in a graph $G'$ by removing all edges between $v$ and $w$, and then identifying $v$ and $w$. This is also known as contraction without removing parallel edges. Again, the result of a contraction is generally a multigraph. Given a set $F \subseteq E(G)$ of edges that induce a connected subgraph of $G$ with $a+1$ vertices of which $v$ is one, after contracting all edges of $F$, we say that $a$ vertices were contracted \emph{onto} $v$.

\begin{lemma}[\cite{ChitnisEtAl2012}] \label{lem:esc:universe}
Given a universe $\mathbb{U}$ and integers $a',b'$ with $0 \leq a',b' \leq |\mathbb{U}|$, one can in time $2^{O(\min\{a',b'\} \log(a'+b'))}\ |\mathbb{U}|\, \log\, |\mathbb{U}|$ find a family $\mc{F}$ of $2^{O(\min\{a',b'\}\log(a'+b'))}\ \log\, |\mathbb{U}|$ subsets of $\mathbb{U}$ such that for any two disjoint sets $A,B \subseteq \mathbb{U}$ of size at most $a'$ and $b'$ respectively, there exists a set $S \in \mc{F}$ that contains all elements of $A$ but is disjoint from $B$.
\end{lemma}

\begin{definition}[\cite{ChitnisEtAl2012}] \label{def:esc:sep}
Given two integers $a,b$, an \emph{$(a,b)$-good edge separation} of a connected graph $G$ is a partition $(V_{1},V_{2})$ of $V(G)$ such that $|V_{1}|, |V_{2}| > a$, $G[V_{1}]$ and $G[V_{2}]$ are connected, and the number of edges between $V_{1}$ and $V_{2}$ is at most $b$.
\end{definition}

\begin{lemma}[\cite{ChitnisEtAl2012}] \label{lem:esc:sep}
Given a connected graph $G$ and two integers $a,b$, one can decide in time $2^{O(\min\{a,b\}\, \log(a+b))}\ |V(G)|^{3}\ \log\, |V(G)|$ whether $G$ has an $(a,b)$-good edge separation, and if it does, find such a separation in the same time.
\end{lemma}

We now define several notions and prove a few lemmas that are implicit in the work of Chitnis~\etal\cite{ChitnisEtAl2012}. We provide full proofs only for sake of completeness.

\begin{definition}
A \emph{$b$-bordered} subgraph of $G$ is a connected induced subgraph $G'$ of $G$ such that in $G$ at most $b$ vertices of $V(G')$ have an edge to a vertex of $V(G) \setminus V(G')$. We call the vertices of $G'$ that have an edge in $G$ to a vertex of $V(G) \setminus V(G')$ the \emph{border vertices} of $G'$.
\end{definition}

\begin{lemma} \label{lem:esc:border}
Given a connected graph $G$ and two integers $a,b$ ($b$ even), one can find in time $2^{O(\min\{a,b\}\, \log(a+b))}\ |V(G)|^{4}\ \log |V(G)|$ a $b$-bordered subgraph of $G$ that does not admit an $(a,b/2)$-good edge separation.
\end{lemma}
\begin{proof}
Initially, let $G' = G$. We apply an iterative procedure that maintains the invariant that $G'$ is a $b$-bordered subgraph of $G$. Note that the invariant holds for $G' = G$. Now run the algorithm of Lemma~\ref{lem:esc:sep} on $G'$ to decide the existence of an $(a,b/2)$-good edge separation in $G'$. If no such separation exists, then we simply return $G'$. Otherwise, let $(V_{1},V_{2})$ be the separation returned by the algorithm. Assume without loss of generality that $V_{1}$ contains less border vertices of $G'$ than $V_{2}$. Since $G'$ is $b$-bordered, $V_{1}$ contains at most $b/2$ border vertices of $G'$. Moreover, since $(V_{1},V_{2})$ is an $(a,b/2)$-good edge separation in $G'$, there are at most $b/2$ edges between $V_{1}$ and $V_{2}$, and thus at most $b/2$ vertices of $V_{1}$ have an edge to a vertex of $V_{2}$ in $G'$. Since $G'$ is an induced subgraph of $G$, this implies at most $b/2$ vertices of $V_{1}$ have an edge to a vertex of $V_{2}$ in $G$. Hence, at most $b$ vertices of $V_{1}$ have an edge to a vertex of $V(G) \setminus V_{1}$. Hence, $G[V_{1}]$ is a $b$-bordered subgraph of $G$. Set $G' = G[V_{1}]$ and iterate. Since $V_{1}$ and $V_{2}$ are both nonempty, this procedure terminates after $O(|V(G)|)$ steps. Combined with the running time of the algorithm of Lemma~\ref{lem:esc:sep}, this implies the lemma.
\end{proof}

Let $G$ be a connected graph and let $a$ be an integer. Given a set $F \subseteq E(G)$, let $G_{F}$ denote the graph obtained from $G$ by contracting all edges of $F$, and then identifying into a single vertex (which we denote by $h_{F}$) all vertices onto which at least $a$ vertices were contracted.
Observe that $G_{F}$ is potentially a multigraph, and that $h_{F}$ might not exist.

A subset $Y$ of the edges of a connected graph $G$ is a \emph{separator} if $G \setminus Y$ has more than one connected component. The set $Y$ is a \emph{minimal} separator if there is no $Y' \subset Y$ such that $G \setminus Y'$ has the same connected components as $G \setminus Y$.

\begin{lemma}[\cite{ChitnisEtAl2012}] \label{lem:esc:nosep}
Let $G$ be a connected graph, let $a,b$ be two integers ($b$ even), and let $F \subseteq E(G)$ with $|F| \leq b/2$. If $G$ admits no $(a,b/2)$-good edge separation, then $G \setminus F$ has at most $(b/2)+1$ connected components, of which at most one has more than $a$ vertices.
\end{lemma}

\begin{lemma} \label{lem:esc:family}
Let $G$ be a connected graph, let $a,b$ be any two integers ($b$ even) such that $G$ does not admit an $(a,b/2)$-good edge separation and such that $|V(G)| > a(b/2+1)$, and let $Y \subseteq E(G)$ with $|Y| \leq b/2$ be a minimal separator.
In time $2^{O(b \log(a+b))}\ |E(G)|\, \log\, |E(G)|$ one can find a family $\mc{F}$ of $2^{O(b \log(a+b))}\ \log\, |E(G)|$ subsets of $E(G)$ that contains a set $F_{0} \subseteq E(G)$ with the following properties: (1) $F_{0} \cap Y = \emptyset$, (2) $h_{F_{0}}$ exists in $G_{F_{0}}$, (3) $h_{F_{0}}$ is the identification of a subset of the vertices of a connected component of $G \setminus Y$, and (4) for each connected component $C$ of $G_{F_{0}} \setminus \{h_{F_{0}}\}$, $Y$ either contains all edges of $G_{F_{0}}[C \cup \{h_{F_{0}}\}]$ or none of these edges.
\end{lemma}
\begin{proof}
The idea of the proof is to show that $Y$ induces a set $F_{0} \subseteq E(G)$ as in the lemma statement, which one can discover among the sets in the family $\mc{F}$ returned by Lemma~\ref{lem:esc:universe} for an appropriate choice of $\mathbb{U}$, $a'$, and $b'$.

We first describe an essential property of $F_{0}$.
Let $C_{0},\ldots,C_{\ell}$ be the connected components of $G \setminus Y$. By Lemma~\ref{lem:esc:nosep} and since $Y$ is a separator, $1 \leq \ell \leq b/2$. Following the same lemma and the assumption that $|V(G)| > a(b/2 + 1)$, there is exactly one component of more than $a$ vertices; without loss of generality, this is $C_{0}$. For every connected component $C_{i}$, let $T_{i}$ denote an arbitrary spanning tree of it. 
For each vertex $v$ of $C_{0}$ incident to an edge of $Y$, let $T_{0}^{v}$ be a subtree of $T_{0}$ that contains $v$ and that has $a+1$ vertices; note that such a subtree exists, as $T_{0}$ has more than $a$ vertices. Let $T_{0}'$ denote the union of all of these subtrees; note that $T_{0}'$ is a forest in general. Then let $A = \bigcup_{i=1}^{\ell} E(T_{i}) \cup E(T_{0}')$ and let $B = Y$.  Let $F_{0}$ be an arbitrary subset of $E(G)$ that contains $A$ but is disjoint from $B$.

We claim that $F_{0}$ satisfies all the properties of the lemma statement. Clearly, $F_{0} \cap Y = F_{0} \cap B = \emptyset$, and thus the first property holds.

Let $H$ denote the graph obtained from $G$ by contracting all edges of $F_{0}$. Since $E(T_{0}') \subseteq A \subseteq F_{0}$ and $T_{0}'$ contains at least one tree that is a subtree of the spanning tree $T_{0}$ and that has $a+1$ vertices and (thus) $a$ edges, there is at least one vertex onto which at least $a$ vertices were contracted. Hence, $h_{F_{0}}$ exists, and thus the second property holds.

Observe that $F_{0}$ does not contain any edges between $C_{i}$ and $G \setminus C_{i}$, as each such edge belongs to $Y$ and $F_{0} \cap Y = F_{0} \cap B = \emptyset$. Therefore, since $\bigcup_{i=1}^{\ell} E(T_{i}) \subseteq F_{0}$, for $i = 1,\ldots, \ell$, $C_{i}$ gets contracted onto a single vertex of $H$, which we denote by $c_{i}$. Since $C_{i}$ has at most $a$ vertices, at most $a-1$ vertices are contracted onto $c_{i}$. Hence, $c_{1},\ldots,c_{\ell}$ exist in $G_{F_{0}}$; in other words, they are not identified with $h_{F_{0}}$. It follows that $h_{F_{0}}$ is the identification of a subset of the vertices of $C_{0}$, a connected component of $G \setminus Y$, and thus the third property holds.

From the above observation, it follows that all edges incident to $c_{i}$ in $G_{F_{0}}$ for $i=1,\ldots,\ell$ belong to $Y$, and since $Y$ is minimal, all edges of $Y$ are incident to a $c_{i}$ in $G_{F_{0}}$ for $i=1,\ldots,\ell$. It remains to show that all edges of $G_{F_{0}}$ incident to exactly one $c_{i}$ for $i=1,\ldots,\ell$ have $h_{F_{0}}$ as their other endpoint. Let $v$ be any vertex of $C_{0}$ incident to an edge $e$ of $Y$. By construction, $v$ is contained in a subtree of $T_{0}'$ with at least $a+1$ vertices. Since $E(T_{0}') \subseteq F_{0}$, $v$ is contained in a connected component of $G[F_{0}]$ with at least $a+1$ vertices. Therefore, to obtain $H$, $v$ is contracted onto a vertex $w$ onto which at least $a$ vertices are contracted. To obtain $G_{F_{0}}$, $w$ is identified with zero or more other vertices to form $h_{F_{0}}$. Since $e\not\in F_{0}$, $e$ is incident on $h_{F_{0}}$.
Therefore, for each connected component $C$ of $G_{F_{0}} \setminus \{h_{F_{0}}\}$, $Y$ either contains all edges of $G_{F_{0}}[C \cup \{h_{F_{0}}\}]$ or none of these edges, and thus the fourth property holds.

Observe that since $|Y| \leq b/2$, $T_{0}'$ is built from at most $b/2$ subtrees of $T_{0}$, each of which has $a+1$ vertices and $a$ edges. Hence, $|A| \leq (a-1)\ell + a(b/2) \leq (2a-1)(b/2)$. Since $|B| = |Y| \leq b/2$, $\mc{F}$ contains a set $F_{0} \subseteq E(G)$ that contains $A$ but is disjoint from $B$, where $\mc{F}$ is the family returned by Lemma~\ref{lem:esc:universe} for $\mathbb{U} = E(G)$, $a' = (2a-1)(b/2)$, and $b'=b/2$. The lemma follows.
\end{proof}
It is important to observe that the construction of the family $\mc{F}$ does not require knowledge of $Y$ itself, beyond that it has size at most $b/2$. Moreover, note that all edges of $Y$ are present in $G_{F_{0}}$, as $F_{0} \cap Y = \emptyset$.

We are now ready to describe the algorithm for \problemESC{} for the parameter $k+t$. The basic intuition is to find a part of the graph that does not have a $(q,k)$-good edge separation for some $q$, but that only has a small number of border vertices. In this part of the graph we find and contract a set of edges that are provably not part of some smallest edge Steiner multicut. We repeat this procedure until the graph is small enough to be handled by an exhaustive enumeration algorithm.

Consider an instance $(G,\mc{T},k)$ with $\mc{T} = \{T_{1},\ldots,T_{t}\}$ of \problemESC{}. We may assume that $G$ is connected. Let $q$ be an integer determined later 
($q$ will depend on $k$ and $t$ only). 
We assume that $|E(G)| > q$, or we can use exhaustive enumeration to solve the problem in $t\, q^{O(k)}$ time.
 
We apply the algorithm of Lemma~\ref{lem:esc:border} to find a $2k$-bordered subgraph $G'$ of $G$ that does not admit a $(q,k)$-good edge separation. Let $B$ denote the set of border vertices of $G'$. Note that possibly $G' = G$, in which case $B = \emptyset$. The idea is now to determine a set of edges of $G'$ that is not used by some optimal solution.

Let $S$ be a smallest edge Steiner multicut of $(G,\mc{T})$. Let $G''$ denote the graph $(B \cup (V(G) \setminus V(G')), E(G) \setminus E(G'))$. Observe that $E(G')$ and $E(G'')$ partition $E(G)$. Let $S' = S \cap E(G')$ and let $S'' = S \cap E(G'')$.
We call a terminal set \emph{active} if it is not separated in $G \setminus S''$.
We call border vertices $b,b' \in B$ \emph{paired} if there is a path between $b$ and $b'$ in $G'' \setminus S''$. Note that this defines an equivalence relation on $B$.
We call an equivalence class $B'$ of this relation \emph{$i$-active} if the terminal set $T_{i}$ is active and the component of $G'' \setminus S''$ that contains $B'$ contains a terminal of $T_{i}$.
Intuitively, this information suffices to compute a set $Z \subseteq E(G')$ such that $Z \cup S''$ is a smallest edge Steiner multicut of $(G,\mc{T})$. Then we could contract (in $G$) all other edges of $G'$ to get a smaller instance, and repeat this until the instance is small enough to be solved by exhaustive enumeration.

Of course, we do not know $S$, and thus we do not know this equivalence relation on $B$ nor which classes are $i$-active for each $i=1,\ldots,t$. However, we can branch over all possibilities.
In each branch, we find a small set of edges, which we mark. At the end, we contract (in $G$) the set of edges of $G'$ that were not marked, and thus reduce the size of the instance. We will prove that in one of the branches, we mark a smallest set $Z$ of edges (of size at most $k$) such that $(Z \cup S'')$ is a smallest edge Steiner multicut (of size at most $k$) of $(G,\mc{T})$. Therefore, after contraction, a smallest edge Steiner multicut (of size at most $k$) of $(G,\mc{T})$ persists (if such a cut existed in the first place). In particular, we argue that we mark $Z$ in the branch with the active classes $\mc{T}^{S}$, the equivalence relation $\mc{B}^{S}$, and $i$-active classes $\mc{B}^{S}_{i}$ of $\mc{B}^{S}$ for $i=1,\ldots,|\mc{T}^{S}|$ that are induced by $S$.

The algorithm now branches over all possibilities. Let $\mc{T}' = \{T_{1}',\ldots,T_{t'}'\}$ be an arbitrary subset of $\mc{T}$, let $\mc{B}$ be an arbitrary equivalence relation on $B$, and let $\mc{B}_{i}$ denote an arbitrary subset of $\mc{B}$ for $i=1,\ldots,t'$. We say that we made the \emph{right choice} if $\mc{T}' = \mc{T}^{S}$, $\mc{B} = \mc{B}^{S}$, and $\mc{B}_{i} = \mc{B}_{i}^{S}$ for $i=1,\ldots,t'$. The algorithm considers two cases.

\medskip\noindent\textbf{Case 1:}\ 
If $|V(G')| \leq q(k+1)$, then we can essentially use exhaustive enumeration.
Let $\tG$ be the graph obtained from $G'$ by identifying two border vertices if they are in the same equivalence class of $\mc{B}$. This also makes each $\mc{B}_{i}$  a set of vertices, which by abuse of notation, we denote by $\mc{B}_{i}$ as well. For $i=1,\ldots,t'$, let $\tT_{i}$ be equal to $\mc{B}_{i} \cup (T_{i}' \cap (V(G') \setminus B))$. Then $\mc{\tT} = \{\tT_{1},\ldots,\tT_{t'}\}$. 
We verify that no terminal set in $\mc{\tT}$ is a singleton; otherwise, we can continue with the next branch.

\begin{lemma} \label{lem:esc:cor1}
Assume we made the right choice. Then $S'$ is an edge Steiner multicut of $(\tG,\mc{\tT})$. Moreover, for any edge Steiner multicut $X$ of $(\tG,\mc{\tT})$, $X \cup S''$ is an edge Steiner multicut of $(G,\mc{T})$.
\end{lemma}
\begin{proof}
Let $v,w$ be any two vertices of $\tG$ (possibly $v=w$) that are in the same connected component of $\tG \setminus S'$, and let $P$ be an arbitrary $v,w$-path in $\tG \setminus S'$. Observe that each border vertex of $\tG$ corresponds to a set of border vertices in $G'$, which are in the same connected component of $G'' \setminus S''$, since we made the right choice. Hence, $P$ can be `expanded' into a $v,w$-walk of $G \setminus S$, and thus $v,w$ are in the same connected component of $G \setminus S$.

Suppose that $\tT_{i}$ is not separated in $\tG \setminus S'$. Consider any two terminals $s,s' \in T_{i}'$. We show that $s$ and $s'$ are in the same connected component of $G \setminus S$. Then there are several cases:
\begin{itemize}
\item if $s,s' \in V(G'')$, then let $b$ and $b'$ be border vertices reachable in $G'' \setminus S''$ from $s$ and $s'$ respectively (note that possibly $b=b'$). By construction, $b,b' \in \tT_{i}$, and since $\tT_{i}$ is not separated in $\tG \setminus S'$, there is a path between $b$ and $b'$ in $\tG \setminus S'$. By the above observation, there is a path between $b$ and $b'$ in $G \setminus S$. Since $b$ and $b'$ are reachable in $G'' \setminus S''$ from $s$ and $s'$ respectively, it follows that $s$ and $s'$ are in the same connected component of $G \setminus S$.
\item if $s,s' \in V(G')$, then $s,s' \in \tT_{i}$. Hence, $s,s'$ are in the same connected component of $\tG \setminus S'$, and by the above observation, in the same connected component of $G \setminus S$.
\item if $s \in V(G')$ and $s' \in V(G'')$, then by combining the reasoning of the above two cases, we can again show that $s$ and $s'$ are in the same connected component of $G \setminus S$.
\end{itemize}
It follows that all terminals of $T_{i}'$ are in the same connected component of $G \setminus S$, contradicting that $S$ is an edge Steiner multicut of $(G,\mc{T})$. Therefore, $S'$ is an edge Steiner multicut of $(\tG,\mc{\tT})$.

Let $X$ be any edge Steiner multicut of $(\tG,\mc{\tT})$. 
Suppose that $T_{i}'$ is not separated in $G \setminus (X \cup S'')$. Consider any two terminals $s,s' \in \tT_{i}$. Define $z\in V(G)$ such that $z=s$ if $s$ is not a border vertex, or such that $z$ is a terminal of $T_{i}'$ in the connected component of $G''\setminus S''$ that contains $s$ otherwise. Define $z'$ similarly. Since we made the right choice and by construction, $z$ and $z'$ are properly defined. Since $T_{i}'$ is not separated in $G \setminus (X \cup S'')$, there is a path $P$ in $G \setminus (X \cup S'')$ between $z$ and $z'$. If $s$ (respectively $s'$) is a border vertex, then $P$ contains $s$ (respectively $s'$). Hence, there is a path $P'$ between $s$ and $s'$ in $G \setminus (X \cup S'')$. Since we made the right choice, any subpath of $P'$ that lies in $G'' \setminus S''$ and goes between two border vertices, can be replaced by the vertex into which these two border vertices were identified. Hence, $P'$ can be `compressed' into an $s,s'$-walk of $\tG \setminus X$, and thus $s,s'$ are in the same connected component of $\tG \setminus X$. It follows that all terminals of $\tT_{i}$ are in the same connected component of $\tG \setminus X$, contradicting that $X$ is an edge Steiner multicut of $(\tG,\mc{\tT})$. Therefore, $X \cup S''$ is an edge Steiner multicut of $(G,\mc{T})$.\end{proof}
Now the algorithm uses exhaustive enumeration to find a smallest edge Steiner multicut of $(\tG,\tT)$ of size at most $k$ (if one exists) in $t(qk)^{O(k)}$ time. Mark this set of edges in $G'$.

By Lemma~\ref{lem:esc:cor1}, if we made the right choice, we mark a set $Z$ such that $Z \cup S''$ is a smallest edge Steiner multicut of $(G,\mc{T})$.

\medskip\noindent\textbf{Case 2:}\ 
If $|V(G')| > q(k+1)$, then a more complicated approach is needed,
because we cannot just use exhaustive enumeration. 
In fact, we cannot work with $(\tG,\mc{\tT})$ directly here, as $\tG$ might have a $(q,k)$-good edge separation, even though $G'$ does not.

We proceed as follows. Apply Lemma~\ref{lem:esc:family} with $a = q$ and $b = 2k$ to $G'$ with respect to $Y := S'$, and let $\mc{F}$ be the resulting family. Note that $Y=S'$ is a minimal separator, $|V(G')| > q(k+1) = a(b/2 + 1)$, and $G'$ has no $(a,b/2)$-good edge separation, and thus the lemma indeed applies. Consider an arbitrary $F \in \mc{F}$. We augment our definition of the \emph{right choice} by adding the condition that $F = F_{0}$, where $F_{0}$ is the family that Lemma~\ref{lem:esc:family} promises exists in $\mc{F}$. Now find $G'_{F}$. If $h_{F}$ does not exist in $G'_{F}$, then we proceed to the next set $F$, as Lemma~\ref{lem:esc:family} promises that $h_{F}$ exists if we made the right choice.

We call a set $X \subseteq E(G'_{F})$ an \emph{all-or-nothing cut} if for each connected component $C'$ of $G'_{F} \setminus \{h_{F}\}$, $X$ either contains all edges of $G'_{F}[C' \cup \{h_{F}\}]$ or none of these edges. Note that Lemma~\ref{lem:esc:family} promises that $S'$ is an all-or-nothing cut in $G'_{F}$ for $F = F_{0} \in \mc{F}$.

Let $\mc{C}$ denote the set of connected components of $G'_{F} \setminus \{h_{F}\}$ that contain a vertex onto which a border vertex was contracted. Let $Y \subseteq \bigcup_{C \in \mc{C}} E(G'_{F}[C \cup \{h_{F}\}])$ be an arbitrary set of edges that contains for each $C \in \mc{C}$ either all edges of $G'_{F}[C \cup \{h_{F}\}]$ or none of these edges. Note that $Y$ is basically an all-or-nothing cut restricted to the edges induced by $\mc{C}$. The algorithm will consider all possible choices of $Y$. We augment our definition of the \emph{right choice} again, by adding the condition that $Y = Y_{0}$, where $Y_{0} := S' \cap (\bigcup_{C \in \mc{C}} E(G'_{F}[C \cup \{h_{F}\}]))$.

Now the algorithm deletes every edge in $Y$ and contracts every edge in $(\bigcup_{C \in \mc{C}} E(G'_{F}[C \cup \{h_{F}\}])) \setminus Y$. Denote the resulting graph by $H_{F}$. Let $\hG$ be the graph obtained from $H_{F}$ by identifying two border vertices if they are in the same equivalence class of $\mc{B}$. This also compresses each $\mc{B}_{i}$ into a set of vertices, which by abuse of notation, we denote by $\mc{B}_{i}$ as well. For $i=1,\ldots,t'$, let $\hT_{i}$ be equal to $\mc{B}_{i} \cup (T_{i}' \cap (V(\hG) \setminus B))$. Then $\mc{\hT}$ consists of all $\hT_{i}$ that are not already separated in $H_{F}$.
We verify that no terminal set in $\mc{\hT}$ is a singleton; otherwise, we can continue with the next branch.

\begin{lemma} \label{lem:esc:cor2}
Assume that we made the right choice. Then $S' \setminus Y$ is an edge Steiner multicut of $(\hG,\mc{\hT})$ that is an all-or-nothing cut. Moreover, for any edge Steiner multicut $X$ of $(\hG,\mc{\hT})$, $Y \cup X \cup S''$ is an edge Steiner multicut of $(G,\mc{T})$.
\end{lemma}
\begin{proof}
Since $Y$ is an all-or-nothing cut of $G'_{F}$ and we made the right choice, it is immediate from Lemma~\ref{lem:esc:family} that $S'\setminus Y$ is an all-or-nothing cut of $G'_{F}$ as well. Since $(\bigcup_{C \in \mc{C}} E(G'_{F}[C \cup \{h_{F}\}])) \setminus Y$ is also an all-or-nothing cut of $G'_{F}$, which is disjoint from $S'$ (as we made the right choice), $S' \setminus Y$ is an all-or-nothing cut of $H_{F}$. Notice that border vertices of $H_{F}$ are either isolated or equal to $h_{F}$. By the construction of $\hG$, $S'\setminus Y$ is an all-or-nothing cut of $\hG$.

The remainder of the proof is similar to Lemma~\ref{lem:esc:cor1}, but more involved due to the more complex construction of $\hG$.

Let $v,w$ be any two vertices of $\hG$ (possibly $v=w$) that are in the same connected component of $\hG \setminus (S' \setminus Y)$, and let $P$ be an arbitrary $v,w$-path in $\hG \setminus (S' \setminus Y)$. Observe that each border vertex of $\hG$ corresponds to a set of border vertices in $H_{F}$, which are in the same connected component of $G'' \setminus S''$, since we made the right choice. Hence, $P$ can be `expanded' into a $v,w$-walk $W$ of $(G'' \setminus S'') \cup (H_{F} \setminus (S' \setminus Y))$. By the construction of $H_{F}$, $W$ can be `expanded' into a $v,w$-walk $W'$ of $(G'' \setminus S'') \cup (G'_{F} \setminus S')$. Because we made the right choice, by Lemma~\ref{lem:esc:family} all vertices identified into $h_{F}$ come from the same connected component of $G' \setminus S'$ and $F \cap S' = \emptyset$, and thus $W'$ can be `expanded' into a $v,w$-walk of $(G'' \setminus S'') \cup (G' \setminus S')$. Hence, $v,w$ are in the same connected component of $G \setminus S$.

Suppose that $\hT_{i}$ is not separated in $\hG \setminus (S' \setminus Y)$. Consider any two terminals $s,s' \in T_{i}'$. We show that $s$ and $s'$ are in the same connected component of $G \setminus S$. There are several cases:
\begin{itemize}
\item if $s,s' \in V(G'')$, then let $b$ and $b'$ be border vertices reachable in $G'' \setminus S''$ from $s$ and $s'$ respectively (note that possibly $b=b'$). By construction, $b,b' \in \hT_{i}$, and since $\hT_{i}$ is not separated in $\hG \setminus (S' \setminus Y)$, there is a path between $b$ and $b'$ in $\hG \setminus (S'\setminus Y)$. By the above observation, there is a path between $b$ and $b'$ in $G \setminus S$. Since $b$ and $b'$ are reachable in $G'' \setminus S''$ from $s$ and $s'$ respectively, it follows that $s$ and $s'$ are in the same connected component of $G \setminus S$.
\item if $s,s' \in V(G')$, then $s,s' \in \hT_{i}$. Hence, $s,s'$ are in the same connected component of $\hG \setminus (S' \setminus Y)$, and by the above observation, in the same connected component of $G \setminus S$.
\item if $s \in V(G')$ and $s' \in V(G'')$, then by combining the reasoning of the above two cases, we can again show that $s$ and $s'$ are in the same connected component of $G \setminus S$.
\end{itemize}
It follows that all terminals of $T_{i}'$ are in the same connected component of $G \setminus S$, contradicting that $S$ is an edge Steiner multicut of $(G,\mc{T})$. Therefore, $S' \setminus Y$ is an edge Steiner multicut of $(\hG,\mc{\hT})$.

Let $X$ be any edge Steiner multicut of $(\hG,\mc{\hT})$. 
Suppose that $T_{i}'$ is not separated in $G \setminus (Y \cup X \cup S'')$. Consider any two terminals $s,s' \in \hT_{i}$. Define $z \in V(G)$ such that $z=s$ if $s$ is not a border vertex, and such that $z$ is a terminal of $T_{i}'$ in the connected component of $G''\setminus S''$ that contains $s$ otherwise. Define $z'$ similarly. Since we made the right choice and by construction, $z$ and $z'$ are properly defined. Since $T_{i}'$ is not separated in $G \setminus (Y \cup X \cup S'')$, there is a path $P$ in $G \setminus (Y \cup X \cup S'')$ between $z$ and $z'$. If $s$ (respectively $s'$) is a border vertex, then $P$ contains $s$ (respectively $s'$). Hence, there is a path $P'$ between $s$ and $s'$ in $G \setminus (Y \cup X \cup S'')$. Observe that $P'$ consists of several subpaths in $G' \setminus (Y \cup X)$ and several in $G'' \setminus S''$. Let $Q$ be any maximal subpath of $P'$ in $G' \setminus (Y\cup X)$. Since $G'_{F}$ is obtained from $G'$ by contracting edges and identifying vertices, $Q$ corresponds to a walk $W$ in $G'_{F} \setminus (Y\cup X)$ between the same vertices. By the construction of $H_{F}$, $W$ corresponds to a walk $W'$ in $H_{F} \setminus X$ between the same vertices. Finally, by the construction of $\hG$, $W'$ corresponds to a walk in $\hG \setminus X$ between the same vertices. Since we made the right choice, any subpath of $P'$ that lies in $G''$ and goes between two border vertices, can be replaced by the vertex of $\hG$ into which these two border vertices were identified. Hence, $P'$ can be `compressed' into an $s,s'$-walk of $\hG \setminus X$, and thus $s,s'$ are in the same connected component of $\hG \setminus X$. It follows that all terminals of $\hT_{i}$ are in the same connected component of $\hG \setminus X$, contradicting that $X$ is an edge Steiner multicut of $(\hG,\mc{\hT})$. Therefore, $X \cup S''$ is an edge Steiner multicut of $(G,\mc{T})$.
\end{proof}

We now aim to find a smallest edge Steiner multicut $X$ of $(\hG,\mc{\hT})$ that is an all-or-nothing cut.
Let $\{C'_{1},\ldots,C'_{u}\}$ be the set of connected components of $\hG \setminus \{h_{F}\}$. Let $\mc{\hT}|_{i}$ denote the set of terminal sets in $\mc{\hT}$ that are separated if one removes all edges of $E(\hG[\{h_{F}\} \cup C'_{i}])$ from $G'_{F}$. Define $z[\mc{U},i]$, where $\mc{U} \subseteq \mc{\hT}$ and $1 \leq i \leq u$, as the size of the smallest all-or-nothing cut of the terminal sets in $\mc{U}$ using only edges in or going out of $C'_{1},\ldots,C'_{i}$. Then for any $\mc{U}\subseteq \mc{\hT}$,
$$z[\mc{U},1] = \left\{ \begin{array}{ll}
\infty & \mbox{if $\mc{U} \not\subseteq \mc{\hT}|_{1}$} \\
|E(G'_{F}[\{h_{F}\} \cup C'_{1}])| & \mbox{otherwise (\ie if $\mc{U} \subseteq \mc{\hT}|_{1}$)}
\end{array}\right.$$
and for $i > 1$,
$$z[\mc{U},i] = \min\left\{ z[\mc{U},i-1],\ |E(G'_{F}[\{h_{F}\} \cup C'_{i}])| + z[\mc{U} \setminus \mc{\hT}|_{i}, i-1] \right\}.$$
Note that $z[\mc{\hT},u]$ holds the size of the smallest edge Steiner multicut of $(\hG,\hT)$ that is an all-or-nothing cut (if one exists). Finding the set achieving this smallest size is straightforward from the dynamic-programming table. Finally, over all choices of $F$ and all choices of $Y$, mark in $G'$ the smallest set of edges that was found if it has size at most $k$.

By Lemma~\ref{lem:esc:cor2}, if we made the right choice, we mark a set $Z$ such that $Z \cup S''$ is a smallest edge Steiner multicut of $(G,\mc{T})$.

\medskip
In both cases, let $M$ be the set of marked edges. Now we contract all unmarked edges $E(G') \setminus M$ in $G$. 
Let $\tilde{G}$ denote the resulting graph; note that in general $\tilde{G}$ is a multigraph. Each time we contract an edge between two vertices $u$ and $v$, we replace $u$ and $v$ by $uv$ in all terminal sets in $\mc{T}$. Let $\tilde{\mc{T}}$ denote the resulting set of terminal sets. 

Observe that if a terminal set $\tilde{T}_{i}$ in $\tilde{\mc{T}}$ is a singleton set and $T_{i}$ was not a singleton set in $\mc{T}$, then we can answer ``no''.
Indeed, if $(G,\mc{T},k)$ would be a ``yes''-instance, then from the construction of $M$, for any smallest edge Steiner multicut $S$ of $(G,\mc{T})$ there is a smallest edge Steiner multicut $Z \cup (S \setminus E(G'))$ of $(G,\mc{T})$ such that $Z \subseteq M$. In particular, since $T_{i} \subseteq V(G')$ for it to be contracted to a single vertex, $Z$ (and thus also $M$) is an edge Steiner multicut of $(G',\{T_{i}\})$, which contradicts that all vertices of $T_{i}$ are contracted into a single vertex.

Now it remains to prove that $(G,\mc{T},k)$ is a ``yes''-instance if and only if $(\tilde{G},\tilde{\mc{T}},k)$ is. Let $\tilde{S}$ be an edge Steiner multicut of $(\tilde{G},\tilde{\mc{T}})$ of size at most $k$. Since $\tilde{G}$ is obtained from $G$ by a sequence of edge contractions (without deleting parallel edges), there is a natural mapping from $E(\tilde{G})$ to $E(G)$. Let $S$ denote the set of edges obtained from $\tilde{S}$ through this mapping. If $S$ is not an edge Steiner multicut of $(G,\mc{T})$, then there is a path between in $G \setminus S$ between any two terminals of some terminal set of $\mc{T}$. However, because of the way $S$ was obtained from $\tilde{S}$, any path in $G\setminus S$ corresponds to a walk in $\tilde{G}\setminus \tilde{S}$. Hence, there is a walk in $\tilde{G} \setminus \tilde{S}$ between any two terminals of this terminal set in $\tilde{\mc{T}}$, a contradiction. The converse follows immediately from the construction of $M$, \ie that for any smallest edge Steiner multicut $S$ of $(G,\mc{T})$ there is a smallest edge Steiner multicut $Z \cup (S \setminus E(G'))$ of $(G,\mc{T})$ such that $Z \subseteq M$.

Finally, we analyze the running time of our algorithm. 
Since there are at most $2k$ border vertices and $t$ terminal sets, there are $r = 2^{O(kt\log k)}$ different branches that we consider for $\mc{T}'$, $\mc{B}$, and $\mc{B}_{i}$, and in each we mark at most $k$ edges. Choose $q = rk+1$. Now note that $|E(G')| \geq q$:
\begin{itemize}
\item if $G=G'$, then this follows from the assumption that $|E(G)| > q$.
\item if $G\not=G'$, then $G'$ was obtained after considering multiple $(q,k)$-good separations. Hence, $|V(G')| > q$, and since $G'$ is connected, $|E(G')| \geq q$.
\end{itemize}
Since $q = rk+1$, at least one edge of $G'$ was not marked and thus contracted. Therefore, $|V(G)|$ decreases by at least one, and the entire procedure finishes after at most $|V(G)|$ iterations.

To analyze the running time, observe that the dynamic-programming procedure requires $2^{O(t)}k + O(t|E(G)|)$ time. The family $\mc{F}$ has $2^{O(k^{2}t\log k)}\ \log\, |E(G)|$ sets and can be constructed in $2^{O(k^{2}t\log k)}\ |E(G)|\ \log |E(G)|$ time. Hence, Case~2 runs in $2^{O(k^{2}t\log k)}\ |E(G)|\ \log |E(G)|$ time. Case~1 runs in $2^{O(k^{2}t\log k)}$ time. Since there are $r = 2^{O(kt\log k)}$ different branches for $\mc{T}'$, $\mc{B}$, and $\mc{B}_{i}$ that we consider, and it takes $2^{O(k^{2}t\log k)}\ |V(G)|^{4}\ \log |V(G)|$ time to find a $2k$-bordered subgraph that does not admit a $(q,k)$-good separation, each iteration takes $2^{O(k^{2}t\log k)}\ |V(G)|^{4}\ \log |V(G)|$ time. Since there are at most $|V(G)|$ iterations, the total running time is $2^{O(k^{2}t\log k)}\ |V(G)|^{5}\ \log |V(G)|$. Actually, using the recurrence outlined by Chitnis~\etal\cite{ChitnisEtAl2012}, one can show a bound on the running time of $2^{O(k^{2}t\log k)}\ |V(G)|^{4}\ \log |V(G)|$.

\section{Tractability for Edge Deletion and Parameter \texorpdfstring{$k+t$}{k + t}: Second Proof}
\label{sec:ktedge}

In this section, we give the second proof of Theorem~\ref{thm:intro:k+t-edge} through important separators and treewidth reduction.
Our presentation of the algorithm is optimized for readability, not for the final runtime.

We first show an easy reduction to the following problem variant.

\defparproblem{Multipedal Steiner Multicut}{$k+t + |Y|$}{A connected undirected graph $G$, a set $Y \subseteq V(G)$, terminal sets $\mc{T} = \{T_1,\hdots,T_t\subseteq V(G)\}$ with $T_i \cap Y \ne \emptyset$ for all $i$, an integer $k$.}{Find an \problemESC{} $S$ of $T_1,\ldots,T_t$ of size at most $k$.}

\begin{lemma} \label{thm:k+t-edge-reduction}
  \problemESC{} for the parameter $k+t$ is fixed-parameter tractable if \problemMPSC{} is fixed-parameter tractable.
\end{lemma}
\begin{proof}
  Note that the collection of terminal sets $\mc T = \{T_1\ldots,T_t\}$ has a hitting set $Y$ of size at most $t$ that can be computed in polynomial time (just pick any node $v_i \in T_i$ and set $Y := \{v_1,\ldots,v_t\}$).
  Since $|Y| \le t$ and $t$ is a parameter, this is indeed a parameterized reduction to \problemMPSC{}.
\end{proof}

In the rest of this section, we first show a fixed-parameter algorithm for the special case of \problemMPSC{} where $|Y| = 1$ (in Section~\ref{sec:unipedalsteinercut}), and then generalize this algorithm to \problemMPSC{} (in Section~\ref{sec:bipedalsteinercut}).
We will use the following simple fact.

\begin{lemma}
\label{lem:ScutsTatx}
  Let $S \subseteq E(G)$, $T$ be a terminal set and $x \in T$. Then $S$ cuts $T$ if and only if there is a node $v \in T$ such that $S$ cuts $v$ from $x$.
\end{lemma}
\begin{proof}
  Recall that ``$S$ cuts $T$'' means that $S$ cuts some pair of nodes in $T$. 
  Thus, if there is a node $v \in T$ such that $S$ cuts $v$ from $x$ then clearly $S$ cuts $T$.
  Further, if no node $v \in T$ is cut from $x$ by $S$, then $T$ is a connected set in $G \setminus S$, so that $T$ is not cut by $S$.
\end{proof}

\subsection{Unipedal Steiner Multicut}
\label{sec:unipedalsteinercut}

\newcommand{\ol}{\overline}

We first show how to solve the special case of \problemMPSC{} for $|Y| = 1$.

\defparproblem{\problemUSC{}}{$k+t$}{A connected undirected graph $G$, a node $y \in V(G)$, terminal sets $T_1,\hdots,T_t\subseteq V(G)$ with $y \in T_i$ for all $i$, an integer $k$.}{Find an \problemESC{} $S$ of $T_1,\ldots,T_t$ of size at most $k$.}

\begin{theorem} \label{thm:USCfpt}
  \problemUSC{} is fixed-parameter tractable.
\end{theorem}

Our algorithm for \problemUSC{} heavily relies on the notions of important separators and closest cuts, due to Marx and Razgon~\cite{MarxRazgon2011} and Marx\cite{Marx2011}. 

Let $G$ be an undirected graph and $Y \subseteq V(G)$.
A set $S \subseteq E(G)$ cuts $V(G)$ into two sets: $R^Y(S) = R^Y_G(S)$, the union of components of $G \setminus S$ that contain a node of $Y$ (i.e., the nodes reachable from $Y$ in $G \setminus S$), and $\ol R^Y(S) = \ol R^Y_G(S)$, the union of components of $G \setminus S$ disjoint from $Y$.

\begin{definition}
  A cut $S\subseteq V(G)$ a \emph{$Y$-closest cut} if it is inclusion-wise minimal (i.e., there is no $S' \subset S$ with $R^Y(S) = R^Y(S')$) and there is no $S' \subseteq E(G)$ with $|S'| \le |S|$ and $R^Y(S') \subset R^Y(S)$.
\end{definition}

Let $X,Y \subseteq V(G)$ be two disjoint sets.
A set $S \subseteq E(G)$ is an \emph{$X-Y$ separator} if no node in~$Y$ is reachable from a node in $X$ in $G \setminus S$.
For a node $x$ we also write $x-Y$ separator instead of $\{x\}-Y$ separator and $R^x(S)$ instead of $R^{\{x\}}(S)$.

\begin{definition}[\cite{Marx2006}]
  Let $X,Y \subseteq V(G)$ be two disjoint sets and $S \subseteq E(G)$ an $X-Y$ separator.
  We call $S$ an \emph{important $X-Y$ separator} if it is inclusion-wise minimal and there is no $X-Y$ separator $S'$ with $|S'| \le |S|$ and $R^X(S') \supset R^X(S)$. 
\end{definition}

The most valuable properties of important separators with size at most $\ell$ are that there are not too many of them, and that we can enumerate all of them in $O^*(f(\ell))$ time.
We remark that typically this is proven for the node deletion version of important separators, but a standard construction of the line graph (augmented with appropriate terminal node) shows that the same result also holds for the edge deletion variant that we consider in this section.

\begin{theorem}[\cite{MarxRazgon2011}]
\label{thm:numimpsep}
  Let $X,Y \subseteq V(G)$ be two disjoint sets.
  The number of important $X-Y$ separators of size at most $\ell$ is at most $4^\ell$.
  Furthermore, these important separators can be enumerated in time $O(4^\ell \ell (n+m))$.
\end{theorem}

We now come to the key property of edge separators that we use for our fixed-parameter result.
It shows that any $Y$-closest set is a \emph{disjoint} union of important $x-Y$ separators, where~$x$ may range over all nodes in $V(G) \setminus Y$.

Fix a set $Y \subseteq V(G)$ and integer $\ell$ and let $x \in V(G) \setminus Y$.
We denote by $\cI_x$ the set of important $x-Y$ separators of size at most $\ell$.
Further set $\cI^Y_\ell := \bigcup_{x \in V(G) \setminus Y} \cI_x$.
We denote a \emph{disjoint} union by $\uplus$. 

\begin{lemma}
\label{lem:disjoint}
  Let $S$ be a $Y$-closest cut.
  Then there are $S_1,\ldots,S_m \in \cI^Y_\ell$ with $S = \biguplus_i S_i$ and $\ol R^Y(S) = \biguplus_i \ol R^Y(S_i)$.
\end{lemma}
\begin{proof}
  Let $S$ be a $Y$-closest cut.
  Let $C_1,\ldots,C_m$ be the components of $G \setminus S$ disjoint from $Y$.
  Let $S_i \subseteq S$ be the edges incident to a node in $C_i$ and a node in $R^Y(S)$.
  Since $C_i$ is not reachable from $Y$, the set $S_i$ consists of all $C_i - R^Y(S)$ edges in $G$.
  Clearly, the sets $S_i$ are pairwise disjoint.
  Further, if $S$ contains an edge in $G[R^Y(S)]$ or $G[\ol R^Y(S)]$ then the deletion of this edge does not change $R^Y(S)$, as this would contradict $S$ being inclusion-wise minimal.
  Hence, we have $S = \biguplus_i S_i$.
  Furthermore, this shows that $\ol R^Y(S) = C_1 \uplus \ldots \uplus C_m = \ol R^Y(S_1) \uplus \ldots \uplus \ol R^Y(S_m)$. 
  
  It remains to argue that the $S_i$ are important separators, i.e., that $S_i \in \cI^Y_\ell$ for all $i$.
  Consider any $x \in C_i$.
  Since $S_i$ consists of all $C_i - R^Y(S)$ edges in $G$, $C_i$ is a connected component, and~$S$ contains no edges in $G[\ol R^Y(S)]$, $S_i$ is an inclusion-wise minimal $x-Y$ separator. 
  Assume for the sake of contradiction that $S_i$ is not an important $x-Y$ separator.
  Then there exists a set $S'_i \subset E(G)$ such that $|S'_i| \le |S_i|$ and $R^{x}(S'_i) \supset R^{x}(S_i)$. 
  Note that $R^{x}(S'_i)$ contains a node~$z$ in $N(C_i) \subseteq R^Y(S_i)$ (otherwise $S_i'$ contains $S_i$).
  Further, observe that $z \not\in \ol R^Y(S)$ (otherwise~$z$ is in some component $C_j$ and in $N(C_i)$, so $S$ contains an edge in $G[\ol R^Y(S)]$, contradicting minimality).
  For notational purposes, set $S_j' := S_j$ for $j \ne i$.
  Consider the set $S' := \bigcup_j S_j'$.
  We have $|S'| \le  \sum_j |S'_j| \le |S|$ (this uses that $S$ is the \emph{disjoint} union of the $S_j$).
  Further, we have $\ol R^Y(S') \supseteq \bigcup_j \ol R^Y(S_j') \supset \ol R^Y(S)$, since $z$ is in $\ol R^Y(S'_i) \supseteq R^{x}(S_i')$ but not in $\ol R^Y(S)$.
  This contradicts $S$ being a $Y$-closest cut.
\end{proof}

Observe that the number of $Y$-closest cuts can be huge, e.g., the star with midpoint $y$ and~$n$ outgoing edges has ${n \choose \ell}$ $y$-closest cuts of size $\ell$.
However, the number of important separators is much smaller: Theorem~\ref{thm:numimpsep} shows that $|\cI^Y_\ell| \le 4^\ell n$.
Since Lemma~\ref{lem:disjoint} shows that $Y$-closest cuts are generated by important separators, we can optimize over $Y$-closest cuts, although their number may be huge.
We explain the details of this in the remainder of this section. 

We want to emphasize that, although the above lemma resembles the Pushing Lemma~\cite[Lemma 3.10]{MarxRazgon2011} at first glance, it does not hold for the node deletion variant of closest cuts and important separators (see Marx and Razgon~\cite{MarxRazgon2011} for the analogous definitions).
By inspection of the following example one can see that a closest cut is in general not a \emph{disjoint} union of important separators in the node deletion case.

\tikzset{vertex/.style={minimum size=2mm,circle,fill=black,draw,inner sep=0pt},
         decoration={markings,mark=at position .5 with {\arrow[black,thick]{stealth}}}}
\begin{figure}[htpb]
  \centering
    \begin{tikzpicture}[xscale=0.99, yscale=0.7]
      \node (y) at (0,0) [vertex,label=below:$y$]{};
      \node (s_1) at (-1,1) [vertex,label=left:$s_1$]{};
      \node (s_2) at (0,1) [vertex,label=left:$s_2$]{};
      \node (s_3) at (1,1) [vertex,label=right:$s_3$]{};
      \node (v_1) at (-0.5,2) [vertex,label=left:$v_1$]{};
      \node (v_2) at (0.5,2) [vertex,label=right:$v_2$]{};
      \draw(y)--(s_1);
      \draw(y)--(s_2);
      \draw(y)--(s_3);
      \draw(s_1)--(v_1);
      \draw(s_2)--(v_1);
      \draw(s_2)--(v_2);
      \draw(s_3)--(v_2);
    \end{tikzpicture}
  \caption{Example showing that a closest cut is in general \emph{not} a disjoint union of important separators in the node deletion case.}
  \label{fig:example}
\end{figure}
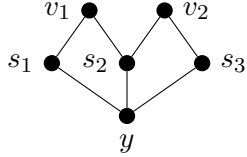

\medskip
\noindent\textbf{Example: Node deletion variant.}
\emph{Consider the graph depicted in Figure~\ref{fig:example}.
Then $\{s_1,s_2,s_3\}$ is a $y$-closest cut, which is the union of the important $v_1-y$ separator $\{s_1,s_2\}$ and the important $v_2-y$ separator $\{s_2,s_3\}$.
However, $\{s_1,s_2,s_3\}$ is not a \emph{disjoint} union of important separators.}
\medskip

We are now ready to present the fixed-parameter algorithm for \problemUSC{}.
Recall that we are given a graph $G$ with a special node $y \in V(G)$ and terminal sets $T_1,\ldots,T_t$ with $y \in T_i$ for all~$i$.
We want to find an \nameESC{} $S^*$ of size at most $k$.
Recall that a set $S \subseteq E(G)$ cuts a terminal set $T_i$ if and only if $S$ cuts a node $v \in T_i$ from $y$ (Lemma~\ref{lem:ScutsTatx}).

We define the \emph{type} of an important separator $S \in \cI^{\{y\}}_k$ as the set of all $i \in \{1,\hdots,t\}$ with $T_i \cap \ol R^{y}(S) \ne \emptyset$, i.e., the set of all $i$ such that $T_i$ is cut by~$S$, denoted by $\type(S)$. 
Our algorithm is a simple dynamic program where, after a trivial initialization, we iterate over all important separators $S \in \cI^{\{y\}}_k$ and update the table $\DP$ as follows:
\begin{algorithmic}
  \algnotext{EndFor}
  \State $\DP[\emptyset] := 0$
  \ForAll{$\emptyset \ne I \subseteq [t]$}
    \State $\DP[I] := \infty$
  \EndFor
  \ForAll{$S \in \cI^{\{y\}}_k$}
    \ForAll{$I \subseteq [t]$}
      \State $\DP[I \cup \type(S)] := \min\{\DP[I \cup \type(S)], \DP[I] + |S| \}$
    \EndFor
  \EndFor
  \Return $\DP[\{1,\ldots,t\}]$
\end{algorithmic}

By Theorem~\ref{thm:numimpsep} we have $|\cI^{\{y\}}_k| \le 4^k n$ and we can enumerate $\cI^{\{y\}}_k$ in $O^*(4^k)$ time.
Overall, the above algorithm runs in time $O^*(2^{2k + t})$.
We remark that the running time of the algorithm is $O(n(n+m))$ for fixed $k,t$.
It remains to prove correctness of the algorithm.

\begin{lemma}
\label{lem:edge:DP}
  The dynamic program returns a value of at most $k$ if and only if the optimal \nameESC{}~$S^*$ has size at most $k$ (in which case the return value coincides with $|S^*|$).
\end{lemma}

Note that determining the size of an optimal solution suffices because of self-reducibility.
With Lemma~\ref{lem:edge:DP}, we also finish the proof of Theorem~\ref{thm:USCfpt}.

\begin{proof}[Proof of Lemma~\ref{lem:edge:DP}]
  Suppose that the optimal \nameESC{}~$S^*$ has size at most $k$.
  We can assume that $S^*$ is a $\{y\}$-closest cut.
  Otherwise, we can replace $S^*$ by a $\{y\}$-closest cut $S'$ with the same (or lower) cost and $R^y(S') \subset R^y(S^*)$, and thus $S'$ still cuts all sets $T_i$ that are cut by $S^*$, see Lemma~\ref{lem:ScutsTatx}.
  By Lemma~\ref{lem:disjoint}, there are important separators $S_1,\ldots,S_m \in \cI^{\{y\}}_k$ such that $S^* = \biguplus_i S_i$, implying that $|S^*| = \sum_i |S_i|$, and $\ol R^y(S^*) = \biguplus_i \ol R^y(S_i)$.
  Note that the latter implies that $\type(S^*) = \bigcup_i \type(S_i)$.
  As $S^*$ cuts all terminal sets $T_i$ we have $\type(S^*) = \{1,\hdots,t\}$, showing $\bigcup_i \type(S_i) = \{1,\hdots,t\}$.
  Hence, $\bigcup_i S_i$ appears as one term in $\DP[\{1,\ldots,t\}]$, so that we return a value $\DP[\{1,\ldots,t\}] \le \sum_i |S_i| = |S^*| \le k$. 
  
  For the other direction, suppose that we return $\DP[\{1,\ldots,t\}] = \ell$.
  Then there are important separators $S_1,\ldots,S_m \in \cI^{\{y\}}_k$ such that $\sum_i |S_i| = \ell$ and $\bigcup_i \type(S_i) = \{1,\hdots,t\}$.
  Consider $S := \bigcup_i S_i$
  Since $\ol R^y(S) \supseteq \bigcup_i \ol R^y(S_i)$ we have $\type(S) \supseteq \bigcup_i \type(S_i) = \{1,\hdots,t\}$, so $S$ is a valid \nameESC{}.
  Further, $|S| \le \sum_i |S_i|$, so we proved the existence of an \nameESC{} of size at most $\DP[\{1,\ldots,t\}]$.
  Together with the first direction this proves the claim. 
\end{proof}

\subsection{Multipedal Steiner Multicut}
\label{sec:bipedalsteinercut}
\newcommand{\torso}{\textup{torso}}

In this section, we show a fixed-parameter algorithm for \problemMPSC{} in general.

We first show that any inclusion-wise minimal \nameESC{} $S$ can be split into two disjoint parts, one part consisting of a union of minimal $x-y$ separators for some $x,y \in Y$, and one part where $Y$ is not cut at all. 
For this, we denote by $Z$ the union of all edges in $G$ that appear in a minimal $x-y$ cut of size at most $k$ for some $x,y \in Y$. 

\begin{lemma} \label{lem:splitmincutunipedal}
  Let $S$ be an inclusion-wise minimal \nameESC{} of $G$. Then we can write $S = S' \uplus S''$ where:
  \begin{enumerate}[label=(\arabic*),noitemsep,topsep=1pt]
    \item $S' \subseteq Z$,
    \item $Y$ is contained in one component of $G \setminus S''$,
    \item $\ol R^y(S) = \ol R^y(S') \cup \ol R^y(S'')$ for all $y \in Y$.
  \end{enumerate}
\end{lemma}
\begin{proof}
  In this proof we denote by $\delta(U)$ the set of edges in $G$ that are incident to a node in $U$ and a node in $V(G) \setminus U$, for any $U \subseteq V(G)$.
  
  Consider an inclusion-wise minimal \nameESC{} $S$ of size at most $k$.
  For $y \in Y$, let $S_y := \delta(R^y(S))$ (i.e., the edges leaving the set of nodes that are reachable from $y$ in $G \setminus S$).
  Observe that $S_y \subseteq S$ and $R^y(S) = R^y(S_y)$.
  By the minimality of $S$, we have $S = \bigcup_{y \in Y} S_y$, since the deletion of any edge of $S$ in $G[R^y(S)]$, for $y \in Y$, or in $G[\ol R^Y(S)]$ is redundant, as its deletion does not change $R^y(S)$ for any $y \in Y$ (and thus does not change which terminal set are disconnected). 
  
  Deleting a set $S_y$ cuts $G$ into some connected components.
  Note that if $C$ is such a component, then $\delta(C) \subseteq S_y$ (otherwise, there is a node $v \in N(C)$ that is reachable in $G \setminus S_y$, which contradicts that $C$ is a component) and $\delta(C)$ consists only of $C - R^y(S_y)$ edges (since~$S_y$ consists of edges incident to $R^y(S_y)$).
  
  Let $C'_1,\ldots,C'_{r+1}$ be the components of $G \setminus S_y$ that contain some node from $Y$.
  Without loss of generality, we have $C'_{r+1} = R^y(S_y)$ (which we ignore from now on) and thus no other $C_i'$ contains $y$.
  For such a component $C_i'$ with $x \in Y \cap C_i'$, $1 \le i \le r$, the set of edges $\delta(C'_i)$ is a minimal $x-y$ cut: since $x \in C_i', y \not \in C_i'$ the set $\delta(C_i')$ is certainly an $x-y$ cut, and since $C_i'$ and $R^y(S_y)$ are connected and $\delta(C_i')$ consists only of $C_i' - R^y(S_y)$ edges, it is even a minimal cut.
  Set $S'_y := \bigcup_{1 \le i \le r} \delta(C_i') \subseteq S$ and $S' := \bigcup_{y \in Y} S'_y \subseteq S$.
  Note that we have $S' \subseteq Z$, since $S'$ is a union of minimal $x-y$ cuts for some $x,y \in Y$, each of size at most $|S| \le k$.
  This proves property (1).
  
  Let $C_1'',\ldots,C_\ell''$ be the remaining components of $G \setminus S_y$, i.e., the components that are disjoint from $Y$, and set $S''_y := \bigcup_i \delta(C_i'') \subseteq S$.
  Further, let $S'' := \bigcup_{y \in Y} S''_y \subseteq S$. 
  We claim that $Y$ is contained in one component of $G \setminus S''$.
  Consider any $x, y \in Y$ and a $x-y$ path $P$ in $G$ that passes through $S''$ a minimum number of times (among all $x-y$ paths).
  Let $z \in Y$ be such that~$S''_z$ contains some edge $e$ of $P$ in $S''$.
  Before $e$, the path $P$ passes through $S'_z$ (if $x \not\in R^z(S)$) or it starts in $R^z(S)$ (if $x \in R^z(S)$).
  Similarly, after $e$, the path $P$ passes through $S'_z$ or it ends in $R^z(S)$.
  In any case, since $R^z(S) \subseteq R^z(S'')$ (as $S'' \subseteq S$) we can replace $P$ by a path that avoids $e$ by routing it through the connected set $R^z(S)$.
  This decreases the number of passes of~$P$ through~$S''$, contradicting minimality.
  This proves property (2).
  
  We show that $S = S' \uplus S''$. First note that $S_y = S'_y \uplus S''_y$ by construction.
  Above we showed $S = \bigcup_{y \in Y} S_y$, which implies $S = S' \cup S''$.
  Moreover, consider any component $C_i''$ of $G \setminus S_y$ (that is disjoint from $Y$, as above).
  Note that $C_i''$ is disjoint from $R^x(S) = R^x(S_x)$ for any $x \in Y$, as $x \not\in C_i''$ and any $C_i'' - x$ path passes through $\delta(C_i'') \subseteq S$.
  This shows that $S''_y \cap S_x = \emptyset$ for any $x,y \in Y, x \ne y$, since any edge $e \in S''_y \cap S_x$ would be incident to all three of the disjoint sets $R^y(S_y) = R^y(S)$, $R^x(S_x) = R^x(S)$, and $C_i''$.
  Hence, $S_y''$ is disjoint from $S_x'$ for all $x \in Y$, so that~$S'$ and $S''$ are disjoint. 
  
  Finally, we show $\ol R^y(S) = \ol R^y(S') \cup \ol R^y(S'')$ for any $y \in Y$.
  With notation as above for the components of $G \setminus S_y$ we have $\ol R^y(S) = \ol R^y(S_y) = C'_1 \cup \ldots \cup C'_r \cup C''_1 \cup \ldots \cup C''_\ell$.
  By construction we have $\ol R^y(S'_y) = C'_1 \cup \ldots \cup C'_r$ and $\ol R^y(S''_y) = C''_1 \cup \ldots \cup C''_\ell$, implying $\ol R^y(S) \subseteq \ol R^y(S') \cup \ol R^y(S'')$.
  Since $S \supseteq S', S''$ we also have $\ol R^y(S) \supseteq \ol R^y(S') \cup \ol R^y(S'')$.
  This proves property (3).
\end{proof}

We may branch for each terminal set whether it is cut by $S'$ or by $S''$, splitting $\mc T$ into $\mc T' \uplus \mc T''$.
Furthermore, we may branch over $k' + k'' = k$ such that $|S'| \le k'$ and $|S''| \le k''$.
This branching leaves us with two subproblems:
\begin{enumerate}[label=(\arabic*),noitemsep]
  \item Find an \nameESC{} $S'$ of $(G,\mc T', k')$ with $S' \subseteq Z$,
  \item Find an \nameESC{} $S''$ of $(G, \mc T'', k'')$ such that $Y$ is contained in one component of $G \setminus S''$.
\end{enumerate}
At first sight, these subproblems are \emph{not independent}, since in Lemma~\ref{lem:splitmincutunipedal} the sets $S'$ and $S''$ are required to be \emph{disjoint}.
Nevertheless, we can solve these subproblems independently and put the solutions together to a solution of the original \problemMPSC{} instance.

However, first we change the above subproblems slightly to make them easier to solve. 
Let $Z' = Z'(G,Y)$ with $Z \subseteq Z' \subseteq V(G)$ to be fixed later.
We do this because it is not known how to determine $Z$, but a superset $Z'$ with nice properties can be found using ideas by Marx et al.~\cite{MarxEtAl2013}, as we will explain later.
Then we replace subproblem (1) by the following problem:

\begin{enumerate}[label=(1')]
  \item Find an \nameESC{} $S'$ of $(G, \mc T', k')$ with $S' \subseteq Z'$.
\end{enumerate}

For subproblem (2), let $G_*$ be the graph obtained from $G$ by adding a new node $y_*$ and connecting $y_*$ to every $y \in Y$ by $k+1$ parallel edges (to avoid parallel edges, one may alternatively subdivide all these parallel edges).
Further, for a terminal set $T$, let $T_* := (T \setminus Y) \cup \{y_*\}$ and let $\mc T_* := \{T_* \mid T \in \mc T\}$.
We replace subproblem (2) by the following problem.
Intuitively, $y_*$ forces $Y$ to be connected, so that actually find an \nameESC{} $S''$ that has $Y$ in one component of $G \setminus S''$.

\begin{enumerate}[label=(2')]
  \item Find an \nameESC{} $S''$ of $(G_*, \mc T_*'', k'')$.
\end{enumerate}

Now we show that the original \problemMPSC{} instance $(G, \mc T, k)$ has a solution if and only if for some branch both subproblems (1') and (2') have a solution, finishing a reduction to the two subproblems.

\begin{lemma}
  There is an \nameESC{} of $(G, \mc T, k)$ if and only if for some branch (over $\mc T' \uplus \mc T'' = \mc T$ and $k' + k'' = k$) there is an \nameESC{} $S'$ of $(G, \mc T', k')$ with $S' \subseteq Z'$ and an \nameESC{} $S''$ of $(G_*, T_*'', k'')$.
\end{lemma}
\begin{proof}
  Consider a branch (over $\mc T' \uplus \mc T'' = \mc T$ and $k' + k'' = k$).
  Assume there is an \nameESC{} $S'$ of $(G, \mc T', k')$ with $S' \subseteq Z'$ and an \nameESC{} $S''$ of $(G_*, T_*'', k'')$.
  We can assume, without loss of generality, that $S''$ does not contain any edges in $E(G_*) \setminus E(G)$, since all these edges have $k+1$ parallel copies, so that picking any of these edges does not separate any nodes. 
  Note that $S := S' \cup S'' \subseteq V(G)$ is of size $|S| \le |S'| + |S''| \le k' + k'' = k$.
  Consider a terminal set $T \in \mc T$ with $y \in T \cap Y$.
  If $T \in \mc T'$, then some node $v \in T$ is cut from $y$ by $S'$ (since if all nodes of $T$ connect to $y$ then $T$ is not cut by $S'$).
  Observe that $\ol R^y(S) \supseteq \ol R^y(S') \cup \ol R^y(S'')$.
  Thus, $v$ is also cut from $y$ by $S$.
  A similar arguments works if $T \in \mc T''$, so that $T_*$ is cut by $S''$.
  Then some node $v \in T_*$ is cut from $y_* \in T_*$.
  Since $y_*$ is connected to $y$ in $G \setminus S''$, $v$ is also cut from $y$ by $S''$.
  Hence, $T$ is cut by $S''$ in $G_*$ and, thus, $T$ is cut by $S$ in $G$.
  We have shown that~$S$ is an \nameESC{} of $(G,\mc T, k)$.
  
  For the other direction, let $S$ be an \nameESC{} of $(G,\mc T, k)$.
  Without loss of generality, we assume that $S$ is inclusion-wise minimal.
  Pick sets $S',S''$ as in Lemma~\ref{lem:splitmincutunipedal} and consider a terminal set $T$ and $y \in T \cap Y$.
  Some node $v \in T$ is cut from $y$ by $S$.
  Since we have $\ol R^y(S) = \ol R^y(S') \cup \ol R^y(S'')$ by property (3) in Lemma~\ref{lem:splitmincutunipedal}, node $v$ is also cut from $y$ by $S'$ or $S''$. Hence, if we let $\mc T' \subseteq \mc T$ be all terminal sets cut by $S'$ and $\mc T'' := \mc T \setminus \mc T'$, then all terminal sets in $\mc T''$ are cut by $S''$. Set $k' := |S'|$ and $k'' := k - k'$, then since $|S'| + |S''| = |S| \le k$ we have $|S''| \le k''$.
  Note that $\mc T' \uplus \mc T'' = \mc T$ and $k' + k'' = k$ is a valid branch.
  Now, $S'$ is an \nameESC{} of $(G, \mc T', k')$ with $S' \subseteq Z \subseteq Z'$ by property (1) in Lemma~\ref{lem:splitmincutunipedal}. 
  Also, $S''$ is an \nameESC{} of $(G,\mc T'', k'')$. 
  To see that $S''$ is an \nameESC{} of $(G_*, \mc T''_*, k'')$, note that by property (2) of Lemma~\ref{lem:splitmincutunipedal}, $Y$ is contained in one connected component of $G \setminus S''$.
  Since in the construction of $G_*$ we only add new paths between nodes of $Y$, $S''$ is an \nameESC{} of $(G_*, \mc T'', k'')$.
  As $y \in Y$ and $y_*$ are connected in $G_* \setminus S''$, and $T \cap Y \ne \emptyset$ in any terminal set $T$, we may replace the nodes in $T \cap Y$ by $y_*$, so $S''$ still cuts $T_*$.
  Hence, $S''$ also cuts $\mc T_*''$.
\end{proof}

Note that $(G_*, \mc T''_*, k'')$ is an instance of \problemUSC{}, since each terminal set contains $y_*$.
Hence, we can solve subproblem (2') in $O^*(f(k,t))$ time by Theorem~\ref{thm:USCfpt}, and particularly, in $O(n(n+m))$ time for fixed $k,t$.

It remains to show how to solve subproblem (1'). 
We want to use the techniques by Marx et al.~\cite{MarxEtAl2013}.
However, as their work considers the node deletion variant of cuts, we first transfer our problem to \problemRNSC{}. 
Let $V := V(G)$, $E := E(G)$ and consider the following graph, which is closely related to the line graph of $G$:
\begin{equation*} 
  H = (E \cup V, \{\{e,e'\} \mid e, e' \in E, |e \cap e'| = 1\} \cup \{\{v,e\} \mid e \in E, v \in V, v \in e\}).
\end{equation*}
Then \problemESC{} on $(G, \mc T', k')$ is equivalent to \problemRNSC{} on $(H, \mc T', k')$, in a sense clarified in Lemma~\ref{lem:equivalenceGH}.

We make use of the following theorem, which is a variant of the Treewidth Reduction Theorem by Marx et al.~\cite[Theorem 2.11]{MarxEtAl2013} (this variant follows from using only the first three lines of the proof of Marx et al.~\cite[Theorem 2.11]{MarxEtAl2013}).
The \emph{torso} of a graph $G$ with respect to a set $U \subseteq V(G)$ is the graph on node set $U$ where $u,v \in U$ are adjacent if and only if there is a $u-v$ path in $G$ that is internally node-disjoint from $U$ (in particular, if $\{u,v\} \in E(G)$, then $u$ and $v$ are adjacent in the torso).
We denote this graph by $\torso(G,U)$.

\begin{theorem}[\cite{MarxEtAl2013}] \label{thm:tw:marx}
  Let $H$ be a graph, $Y \subseteq V(H)$, and let $k$ be an integer.
  Let $\tilde Z$ be the set of nodes of $H$ participating in a minimal $x-y$ cut $\tilde S \subseteq V(H)$ of size at most $k$ for some $x,y \in Y$.
  There is an algorithm that, in time $O^*(f(k,|Y|))$, computes a set $\tilde Z' \supseteq \tilde Z \cup Y$ such that $\torso(H,\tilde Z')$ has treewidth at most $h(k,|Y|)$ for some function $h$.
\end{theorem}
We remark that the algorithm of the above theorem runs in $O(n+m)$ time for fixed $k,|Y|$.

Now we fix $Z' := \tilde Z' \cap E(G)$.
Note that $\tw(\torso(H, Z')) \le \tw(\torso(H, \tilde Z')) \le h(k, |Y|)$.
We augment the graph $\torso(H, Z')$ by the nodes $V(G)$ to a new graph $H'$, as follows.

For each node $v \in V(G)$, consider the connected component $C$ of $H \setminus Z'$ that contains $v$.
Let~$K$ be the set of neighbors of $C$ in $Z'$.
Since $C$ is a connected component and by the construction of the torso, $K$ forms a clique in $\torso(H, Z')$.
We add the node $v$ to the graph $\torso(H, Z')$ and make it adjacent to all elements of $K$.
This yields a new graph $H'_v$.
We show that $H'_v$ also has bounded treewidth.
Since $K$ is a clique, by basic knowledge about tree decompositions, every tree decomposition of $\torso(H, Z')$ has a bag containing $K$~\cite[Lemma 3.1]{BodlaenderMoehring1993}; let $(T,\mathcal B)$ be a minimum-width tree decomposition of $\torso(H, Z')$ and let $B\in\mathcal B$ be a bag containing $K$.
We add a new bag $B' = B \cup \{v\}$ to $(T,\mathcal B)$ and make it adjacent to $B$, yielding a tree decomposition of $H'_v$.
This tree decomposition of $H'_v$ has width at most one more than the width of $(T,\mathcal B)$, and thus the treewidth of $H'_v$ is at most one more than the treewidth of $\torso(H, Z')$.
As the constructions of $H'_v$ and $H'_{v'}$ for distinct nodes $v,v' \in V(G)$ do not interfere, in total the treewidth of $H'$ is increased by at most one compared to the treewidth of $\torso(H,Z')$.
Hence, we have $\tw(H') \le h(k, |Y|) + 1$. 

\begin{lemma}
\label{lem:equivalenceGH}
  There is an \nameESC{} $S$ of $(G, \mc T', k')$ with $S \subseteq Z'$ if and only if there is a \nameRNSC{} $S'$ of $(H', \mc T', k')$.
\end{lemma}
\begin{proof}
  Let $S$ be a \nameRNSC{} of $(H', \mc T', k')$.
  Without loss of generality, we can assume that $S \subseteq Z'$, since $V(H') = Z' \cup V$ and the deletion of any non-terminal node $v \in V$ is redundant, since the neighborhood $N(v)$ is a clique (so that any path through $v$ can be stripped of $v$).
  We show that $S$ is also an \nameESC{} of $(G, \mc T', k')$.
  Assume, for the sake of contradiction, that $u,v \in T \in \mc T'$ are cut by $S$ in $H'$, but there is a $u-v$ path in $G \setminus S$.
  Let $(e_1,\ldots,e_\ell)$ be the edges of this path.
  Then $(u,e_1,\ldots,e_\ell,v)$ is a $u-v$ path in $H \setminus S$ (where this time we wrote down the \emph{nodes} of the path).
  Let $e_i$ be the first and $e_j$ be the last edge of this path that is in $Z'$.  
  By basic properties of the torso~\cite[Proposition 3.3]{MarxRazgon2011}, $e_i$ and $e_j$ are also connected in $\torso(H,Z') \setminus S$.
  Further, by the construction of $H'$, $u$ is adjacent to $e_i$ and $v$ is adjacent to~$e_j$ in $H' \setminus S$.
  Hence, there is a $u-v$ path in $H' \setminus S$, which contradicts that $S$ cuts $u,v$ in $H'$.
  
  For the other direction, let $S$ be an \nameESC{} of $(G, \mc T', k')$ with $S \subseteq Z'$.
  We show that $S$ is also a \nameRNSC{} of $(H', \mc T', k')$. 
  Assume, for the sake of contradiction, that $u,v \in T \in \mc T'$ are cut by $S$ in $G$, but not in $H'$, so that there is a $u-v$ path in $H' \setminus S$.
  By the construction of $\torso(H,Z')$ and $H'$ (and since $S \subseteq Z'$), there is also a $u-v$ path in $H \setminus S$.
  However, any such path corresponds to a path in $G \setminus S$, which contradicts that $u,v$ are cut by $S$ in $G$.
  Hence, $S$ is a \nameNSC{} of $(H', \mc T', k')$. 
  Since $S \subseteq Z' \subseteq E$ and $\mc T' \subseteq V$, $S$ is also a \nameRNSC{}.
\end{proof}

Finally, the instance $(H', \mc T', k')$ can be solved in $O^*(f(k,\mathsf{tw}(G)))$ time, since the treewidth of~$H'$ is bounded in $k$ and \problemRNSC{} is fixed-parameter tractable for parameter $t + \tw$ by Theorem~\ref{thm:t+tw}. In particular, that algorithm runs in $O(n+m)$ time for fixed $t,\tw$.
Therefore, \problemMPSC{} is fixed-parameter tractable. The final algorithm runs in time $O(n(n+m))$ for fixed $k,t$; moreover, since the treewidth of the structure returned by Theorem~\ref{thm:tw:marx} is exponential, the algorithm requires double-exponential time.
Combined with Lemma~\ref{thm:k+t-edge-reduction}, this proves Theorem~\ref{thm:intro:k+t-edge}.

\section{Steiner Multicuts for Graphs of Bounded Treewidth}
\label{sec:steinercutsintreewidth}

In this section, we consider \problemSC{} on graphs of bounded treewidth.
To start the exposition, we note that \problemESC{} and \problemRNSC{} are $\mathsf{NP}$-complete for trees and \problemNSC{} is $\mathsf{NP}$-complete on series-parallel graphs~\cite{CalinescuEtAl2003}, which are graphs of treewidth two.
This means that any efficient algorithm for \problemSC{} on graphs of bounded treewidth needs an additional parameter. 

We first show Theorem~\ref{thm:intro:k}, namely that all variants of \problemSC{} for the parameter $k$ are $\mathsf{W}[1]$-hard, even if $p=3$ and $\tw(G) = 2$ (but $t$ is unbounded).
The graph $G$ is in fact a tree plus one node.
We then contrast this result by showing that the problem is fixed-parameter tractable on bounded treewidth graphs when $t$ is a parameter.

For the reduction to prove Theorem~\ref{thm:intro:k} we introduce an intermediate problem, that we call ({\sc Monotone}) \problemNAE{}.
In this problem, we are given variables $x_1,\ldots,x_k$ that each take a value in $\{1,\ldots,n\}$ and clauses $C_1,\ldots,C_m$ of the form
\begin{equation*}
  \NAE( x_{i_1} \le a_1, x_{i_2} \le a_2, x_{i_3} \le a_3 ),
\end{equation*}
where $a_1,a_2,a_3 \in \{1,\ldots,n\}$,
and such a clause is satisfied if not all three inequalities are true and not all are false (i.e., they are ``not all equal'').
The goal is to find an assignment of the variables that satisfies all given clauses. 
We remark that \problemNAE{} generalizes \textsc{Monotone NAE-3-SAT} (by restriction to $n=2$), and  that \problemNAE{} can be solved in time $O(m \cdot n^k)$, by enumerating all assignments.
We complement this with a $\mathsf{W}[1]$-hardness result for parameter $k$.

To prove that \problemNAE{} is $\mathsf{W}[1]$-hard parameterized by $k$, we reduce from \problemMCC{}.
In that problem, which is known to be $\mathsf{W}[1]$-hard~\cite{FellowsEtAl2009}, we are given a graph $G$ and a (proper) coloring of $G$ using $k$ colors, and the goal is to decide if $G$ has a clique that contains at least one node with each of the $k$ colors.
We use $V_{i}$ to denote the set of nodes of color~$i$, set $n_{i} = |V_{i}|$, and $E_{i,j}$ to denote the set of edges with one endpoint in $V_{i}$ and the other in $V_{j}$.

\begin{lemma}
\label{lem:NAEhard}
  \problemNAE{} is $\mathsf{W}[1]$-hard for parameter $k$.
\end{lemma}
\begin{proof}
  Let $(G,k)$ be an instance of \problemMCC{}.
  We create an instance of \problemNAE{} on variables $x_i$, one for each color $1 \le i \le k$, and $y_{ij}$, one for each pair of colors $1 \le i < j \le k$.
  We identify the nodes $V_{i}$ with the integers $\{1,\ldots,n_i\}$ in an arbitrary way.
  We restrict $x_i$ to $\{1,\ldots,n_i\}$ using the clause $\NAE(x_i \le 0, x_i \le 0, x_i \le n_i)$, and write $x_i = u$ if the number $x_i$ corresponds to node $u$.
  Analogously, we can identify the edges $uv \in E_{i,j}$ with numbers in $\{1,\ldots,|E_{i,j}|\}$ and write $y_{ij} = uv$ if we pick the number corresponding to edge $uv$. 
  Consider the following two constraints (for any edge $uv$, with $u$ of color $i$ and $v$ of color $j$):
  \begin{equation*}
    y_{ij} = uv \;\;\Rightarrow\;\; x_i = u, \qquad~\mbox{and}\qquad y_{ij} = uv \;\;\Rightarrow\;\; x_j = v \enspace .
  \end{equation*}
  If we can encode these constraints with $\NAE$-clauses, then any satisfying assignment of the constructed \problemNAE{} instance corresponds to a clique in $G$, as all chosen pairs $y_{ij}$ correspond to edges, and edges sharing a color $i$ picked the same node $x_i$.
  We focus on the first constraint; the second constrained is handled similarly.
  Note that the first constraint is equivalent to 
  \begin{equation*}  
  y_{ij} = uv \;\;\Rightarrow\;\; x_i \ge u, \qquad~\mbox{and}\qquad y_{ij} = uv \;\;\Rightarrow\;\; x_i \le u \enspace .
  \end{equation*}
  Again, without loss of generality, we focus on the first of these constraints.
  It is equivalent to 
  \begin{equation*}
    y_{ij} < uv \;\;\vee\;\; y_{ij} > uv \;\;\vee\;\; x_i \ge u,
  \end{equation*}
  which in turn can be written as
  \begin{equation*}
    \NAE(y_{ij} < uv,\, y_{ij} > uv,\, x_i \ge u),
  \end{equation*}
  since $y_{ij} < uv,\, y_{ij} > uv$ cannot both be true. Note that we can replace any inequality $x < a$ by $x \le a-1$ (and similarly for $x > a$).   
  Hence, we can encode all desired constraints if we may use ``$\le$'' and ``$\ge$'' inequalities, not only ``$\le$'' inequalities, as is the case in the definition of \problemNAE{}.
  
  In the remainder of this proof, we reduce \problemNAE{} with ``$\le$'' and ``$\ge$'' inequalities to the original variant with only ``$\le$'' inequalities.
  Given any instance of \problemNAE{} with both types of inequalities, for any variable $x$ we introduce a new variable~$\bar x$.
  For any $1 \le v \le n$, we add the constraint $\NAE(x \le v,\, x \le v,\, \bar x \le n-v)$.
  This enforces $\bar{x} = n+1 - x$.
  Finally, we replace any inequality $x \ge v$ by $\bar{x} \le n+1-v$.
  This yields an equivalent \problemNAE{} instance with only ``$\le$'' inequalities.
\end{proof}
We are now ready to prove Theorem~\ref{thm:intro:k}.

\begin{proof}[Proof of Theorem~\ref{thm:intro:k}]
  We first give a reduction from \problemNAE{} to \problemESC{} (satisfying $p=3$ and $\tw(G) = 2$).
  Then we show how to generalize the reduction to \problemRNSC{} and \problemNSC{}.

  Consider an instance of \problemNAE{} on variables $x_1,\ldots,x_k$ taking values in $\{1,\ldots,n\}$ with clauses $C_1,\ldots,C_m$.
  Take $k$ paths consisting of $n$ edges and identify their start nodes (to a common node $s$) and end nodes (to a common node $t$), respectively.
  The resulting graph $G$ has $\tw(G) = 2$, since it is not a tree, but becomes a tree after deleting $s$ (or $t$).
  Let~$v^i_j$ be the $j$-th node on the $i$-th path from $s$ to $t$, so that $v^i_0 = s$ and $v^i_n = t$.
  For each clause $\NAE( x_{i_1} \le a_1, x_{i_2} \le a_2, x_{i_3} \le a_3 )$ we introduce a terminal set $\{v^{i_1}_{a_1}, v^{i_2}_{a_2}, v^{i_3}_{a_3}\}$ (note that we can assume $0 \le a_j \le n$ without loss of generality).
  Further, we let $\{s,t\}$ be a terminal set and set the cut size to $k$, i.e., we allow to delete $k$ edges.
  This finishes the construction.
  In order to separate $s$ from $t$ we need to cut at least one edge of each of the $k$ paths that connect $s$ and~$t$, and because the cut size is $k$ we have to delete \emph{exactly one} edge per path.
  Say we delete the $x_i$-th edge on the $i$-th path.
  This splits $G$ into two components, one containing $s$ and the other containing $t$.
  Note that we separate nodes $v^i_j$ and $v^{i'}_{j'}$ by cutting at $x_i \le j$ and $x_{i'} > j'$ (or with both inequalities the other way round), since then $v^i_j$ is in the $t$-component and $v^{i'}_{j'}$ in the $s$-component.
  Hence, the following are equivalent:
  \begin{itemize}[noitemsep]
    \item the terminal set $\{v^{i_1}_{a_1}, v^{i_2}_{a_2}, v^{i_3}_{a_3}\}$ is disconnected;
    \item some pair of nodes in this set is disconnected;
    \item among the inequalities $x_{i_j} \le a_j$, $j=1,2,3$, one is true and one is false;
    \item the clause $\NAE( x_{i_1} \le a_1, x_{i_2} \le a_2, x_{i_3} \le a_3 )$ is satisfied. 
  \end{itemize}
  Therefore, the given \problemNAE{} instance is equivalent to the constructed \problemESC{} instance. 

  To prove hardness for \problemNSC{}, we adapt the construction for \problemESC{}.
  We take the same graph, but let each path contain $2n-1$ internal nodes.
  As before, denote by $v^i_j$ the $j$-th node on the $i$-th path, so that now $v^i_0 = s$ and $v^i_{2n} = t$.
  We introduce terminal sets $\{v^i_1, v^i_{2n-1}\}$ for all $1 \le i \le k$.
  To separate these sets we have to delete at least one inner node of each of the $k$ paths (i.e., a node among $v^i_1,\ldots,v^i_{2n-1}$).
  By setting the cut size to $k$ we make sure that we delete exactly one inner node of every path.
  Say we delete the $x'_i$-th node $v^i_{x'_i}$ of the $i$-th path.
  We call the node $v^i_{2a-2}$ and $v^i_{2a-1}$ the \emph{representatives} of ``$x_i = a$''.
  Thus, for each clause $\NAE( x_{i_1} \le a_1, x_{i_2} \le a_2, x_{i_3} \le a_3 )$ we introduce the terminal sets $\{v^{i_1}_{2 a_1 - c_1}, v^{i_2}_{2 a_2 - c_2}, v^{i_3}_{2 a_3 - c_3}\}$ for all $c_1,c_2,c_3 \in \{1,2\}$, meaning that for all representatives we need to separate some pair.
  Note that after deleting the $x'_i$-th node, at least one representative of ``$x_i=a$'' survives, and the one or two surviving representatives lie in the same component of $G \setminus \{x'_1,\ldots,x'_k\}$.
  We let $x_i$ be the minimal value $a^*$ such that a surviving representative of ``$x_i=a^*$'' lies in the $t$-component. 
  Then the surviving representatives of ``$x_i = a$'' are contained in the $t$-component if and only if $x_i \le a$.
  Hence, taking a terminal set with surviving representatives $\{v^{i_1}_{2 a_1 - c_1}, v^{i_2}_{2 a_2 - c_2}, v^{i_3}_{2 a_3 - c_3}\}$, the terminal set is satisfied if and only if the clause $\NAE(x_{i_1} \le a_1, x_{i_2} \le a_2, x_{i_3} \le a_3 )$ is satisfied, completing the proof.

  For \problemRNSC{}, hardness now follows from Lemma~\ref{lem:NSCtoRNSC}.
\end{proof}

We now contrast the above theorem by showing that \problemSC{} is fixed-parameter tractable for the parameter $t$ if the graph has bounded treewidth.

\begin{theorem}
\label{thm:t+tw}
  \problemNSC{}, \problemESC{}, and \problemRNSC{} are fixed-parameter tractable for parameter $t+\tw(G)$.
\end{theorem}
\begin{proof}
  We first present an MSOL formula for \problemNSC{}, extending the work of Gottlob and Lee~\cite{GottlobLee2007} and Marx~\etal\cite{MarxEtAl2013} for \problemNMC{}.
  Gottlob and Lee~\cite{GottlobLee2007} construct an MSOL formula $\mathrm{connects}(R,x,y)$, equal to
  \begin{equation*}
    R(x) \wedge R(y)  \wedge \forall P \left( \left( P(x) \wedge \lnot P(y) \right)
    \rightarrow
    \left( \exists v,w \left( R(v) \wedge R(w) \wedge P(v) \wedge \lnot P(w) \wedge \mathrm{adj}(v,w) \right) \right) \right),
  \end{equation*}
  which expresses that $x$ and $y$ are connected in the subgraph of $G$ induced by $R$.
  We then extend the formula of Marx~\etal\cite{MarxEtAl2013} to $p > 2$.
  We construct the following MSOL formula:
  \begin{equation*}
    \bigwedge_{i=1}^{t} \exists x, y \left( T_{i}(x) \wedge T_{i}(y) \wedge \lnot (x = y) \wedge \forall R \left( \mathrm{connects}(R,x,y)
    \rightarrow 
    \exists v \left(X(v) \wedge R(v)\right) \right) \right),
  \end{equation*}
  where $X$ is the potential cut set and $T_{i}(x)$ is true if and only if $x$ is in $T_{i}$.
  A slight modification yields MSOL-formulas for \problemESC{} and \problemRNSC{}.
  Note that $X$ is a free set variable in the formulae.

  To apply these formulae, we use Bodlaender's algorithm~\cite{Bodlaender1996} to find a tree decomposition of~$G$ of width $\tw(G)$ in time $f(\tw(G)) \cdot \poly(|V(G)|)$, for some computable function $f$.
  Then we input the necessary MSOL formula $\phi$ and the tree decomposition into the algorithm of Arnborg~\etal\cite{ArnborgEtAl1991}, which runs in time $g(|\phi|,\tw(G))$ for some function $g$ and finds a smallest set $X$ that satisfies $\phi$.
  Since $|\phi|$ is polynomial in $t$, the theorem follows.
\end{proof}

We remark that one can develop an algorithm that has a single-exponential running time of $O^*( 2^{O(\tw(t + \log \tw))} )$ via dynamic programming. 
Because such an algorithm is a straightforward application of known techniques, we only provide a sketch by defining the entries of the dynamic programming table. 
First we compute an approximately optimal tree decomposition in time $O^*(2^{O(\tw)})$~\cite{bodlaender2013c}. For each bag $B$ in the tree decomposition we have a dynamic programming table with entries as follows. Let $S$ be any set of $k'$ edges, $0 \le k' \le k$, in the graph $G_B$ induced by the subtree below $B$. Let $\mc{T}' \subset \mc{T}$ be the terminal sets that are cut by $S$ even after making $B$ into a clique; these terminal sets are cut no matter what we do above $B$. Let $C_1,\ldots,C_r$ be the connected components of $G_B \setminus S$ and let $P_1,\ldots,P_\ell$ be the induced partitioning of the vertices in $B$. For any $T \in \mc{T} \setminus \mc{T}'$ and any part $P_i$ we store whether a vertex in $T$ is connected to a vertex in $P_i$ in $G_B$. This reachability information together with $k'$, $\mc{T}'$, and $P_1,\ldots,P_\ell$ can be stored using $O(\tw(t + \log \tw))$ bits. Thus, the dynamic programming table has $2^{O(\tw(t + \log \tw))}$ possible entries, and for each possible entry we store whether it is realized by some set of edges $S$. At the root, from this information we can extract whether there exists an edge Steiner multicut.

\section{Hardness for Cutsize \texorpdfstring{$k$}{k} and Number of Terminal Sets \texorpdfstring{$t$}{t}}
\label{sec:kt}

In this section, we consider the \problemSC{} problem on general graphs parameterized by $k+t$.
We show that both node deletion versions of the problem, \problemNSC{} and \problemRNSC{}, are $\mathsf{W}[1]$-hard for this parameter. 

\begin{theorem}
\label{thm:k+t}%
  \problemNSC{} and \problemRNSC{} are $\mathsf{W}[1]$-hard for the parameter $k+t$.
\end{theorem}
\begin{proof}
  We present a parameterized reduction from the \problemMCC{} problem~\cite{FellowsEtAl2009} to \problemNSC{}.
  Let $(H,k)$ be an instance of \problemMCC{}, and let~$V_{i}$ and $E_{i,j}$ be as in the definition of \problemMCC{}.
  We then create the following instance of \problemNSC{}.
  First, we subdivide each edge of $H$, and let $N_{i,j}$ denote the set of nodes that were created when subdividing the edges of $E_{i,j}$.
  Then, add a complete graph $C$ with $2k$ nodes, where we denote the nodes of $C$ by $c_{1},\ldots,c_{2k}$, and make all nodes of $V_{i}$ adjacent to $c_{2i-1}$ and $c_{2i}$ for each $i=1,\ldots,k$.
  Let $G$ denote the resulting graph.
  Observe that $G[V(H)]$ and $G[\bigcup_{i,j} N_{i,j}]$ are both independent sets of $G$.
  We then create terminal sets $T_{i} = V_{i} \cup \{c_{2i-1}\}$ and $T_{i}' = V_{i} \cup \{c_{2i}\}$, and terminal sets $T_{i,j} = N_{i,j} \cup V(C)$. Let $\mc{T} = \{T_{i},T_{i}' \mid i=1,\ldots,k\} \cup \{ T_{i,j} \mid i\not=j, i,j = 1,\ldots,k\}$.
  Then the created instance is $(G,\mc{T},k)$.

  Suppose that $(H,k)$ is a ``yes''-instance of \problemMCC{}, and let $K$ denote a clique of $H$ such that $V(K) \cap V_{i} \not= \emptyset$ for each $i$.   
  Pick a node $v_{i} \in V(K) \cap V_{i}$ and let $S = \{v_{i} \mid i=1,\ldots,k\}$.
  Observe that $v_{i}$ disconnects terminal sets $T_{i}$ and $T_{i}'$.
  Further, if we let $n_{i,j}$ denote the subdivision node of the edge $(v_{i},v_{j}) \in E(H)$, then $v_{i}$ and $v_{j}$ disconnect $n_{i,j}$ from the rest of $T_{i,j}$.
  Finally, $|S| = k$.
  Therefore, $(G,\mc{T},k)$ is a ``yes''-instance of \problemNSC{}.

  Suppose that $(G,\mc{T},k)$ is a ``yes''-instance of \problemNSC{}, and let $S \subseteq V(G)$ denote a \nameNSC{} of $G$ with respect to terminal sets $\mc{T}$ such that $|S| \leq k$.
  We claim that $H[S]$ is a multicolored clique of $H$.
  First, observe that to disconnect $T_{i}$ and $T_{i}'$, we need that $S \cap T_{i} \not= \emptyset$ and $S \cap T_{i}' \not= \emptyset$.
  Since $|S| \leq k$, we know that $S \cap T_{i} \cap T_{i}' \not= \emptyset$ for each $i=1,\ldots,k$.
  This implies that $|S| = k$, that $S \subseteq V(H)$, and that $S \cap V_{i} \not= \emptyset$ for each $i=1,\ldots,k$.
  It remains to show that $H[S]$ is a clique. Let $v_{i}$ denote the node in $S \cap V_{i}$.
  Suppose that $(v_{i},v_{j}) \not\in E(H)$ for some $i,j$.   
  Then consider the terminal set $T_{i,j}$, and observe that for any node $n \in N_{i,j}$ at least one endpoint of the edge corresponding to $n$ is not in $S$.
  Since $S \cap V(C) = \emptyset$, this implies that $T_{i,j}$ is not disconnected by $S$, a contradiction.
  It follows that $H[S]$ is a clique, and thus $(H,k)$ is a ``yes''-instance of \problemMCC{}.

  For \problemRNSC{} hardness follows from the statement for \problemNSC{} and Lemma~\ref{lem:NSCtoRNSC}.
\end{proof}

\section{Steiner Multicuts in Trees}
\label{sec:steinercutsintrees}

In this section, we consider the complexity of \problemSC{} on trees and completely characterize it with respect to various parameters.
Throughout, we consider instances $(G,\mc{T},k)$ of \problemSC{} (in either the edge, node, or restricted-node variant) where $G$ is a tree.
Here we assume that $G$ is rooted at an arbitrary node $r \in V(G)$.
For each terminal set $T \in \mc{T}$, let $G_{T}$ denote the subtree of $G$ induced by the terminals of $T$.
That is, $G_{T}$ is the union of the shortest paths between each pair of terminals of $T$.
Since $G$ is rooted, $G_{T}$ also has a root, denoted by $r(G_{T})$ and called a \emph{terminal root}.

\subsection{Polynomial-time Algorithm for \problemNSC{}}
It is known that \problemNMC{} can be solved in polynomial time on trees~\cite{CalinescuEtAl2003}.
Here, we extend this algorithm to \problemNSC{}.

\begin{theorem} \label{thm:node:easy}
  \problemNSC{} can be decided in linear time on trees.
\end{theorem}
\begin{proof}
  Let $(G,\mc{T},k)$ be an instance of \problemNSC{} where $G$ is a tree.
  Let $r$ and the terminal roots of $G$ be defined as before.
  A crucial observation is that, for any $T \in \mc{T}$, if the subtree $G' \supseteq G_{T}$ of $G$ rooted at $r(G_{T})$ contains no terminal roots except $r(G_{T})$, then there is an optimal solution that contains $r(G_{T})$.
  To see this, note that any solution must contain at least one node of $G'$ in order to disconnect $T$.
  Further, since $G'$ has no terminal roots except $r(G_{T})$, the set of terminal sets disconnected when removing $r(G_{T})$ is a superset of the set of terminal sets disconnected when removing any node of $G'$.

  Using the above crucial observation, a greedy strategy becomes apparent.
  First, if $k=0$ and $\mc{T} \not=\emptyset$, then return ``no''; if $k \geq 0$ and $\mc{T} = \emptyset$, then return ``yes''.
  Otherwise, compute the terminal root of each terminal set.
  Then, find a terminal root $c$ that is deepest in $G$.
  Add $c$ to the cut, and recurse on the instance $(G', \mc{T}', k-1)$, where $G'$ is $G$ minus the subtree rooted at~$c$, and $\mc{T}'$ is $\mc{T}$ minus every terminal set that is disconnected by $c$.

  The above algorithm can be implemented in linear time, as follows.
  First, we find all terminal roots.
  To this end, index the nodes of $G$ according to a post-order traversal.
  Then, for each terminal set $T \in \mc{T}$, its terminal root is the lowest common ancestor of the lowest and highest indexed terminal in $T$.
  We can use the lowest common ancestor data structure of Harel and Tarjan~\cite{HarelTarjan1984} to find all these lowest common ancestors, and thus all terminal roots, in linear time.
  Second, we perform the greedy algorithm.
  For each node of $G$, we maintain a queue with a set of terminals on that node.
  Now apply an inverse breadth-first search; that is, number the node first visited by a breadth-first search (i.e., $r$) by $|V(G)|$, and the last one by $1$, and then visit the nodes starting with number $1$ up to $|V(G)|$.
  Suppose that we visit a node $v$.
  If $v$ is not a terminal root or $v$ is a terminal root for a terminal set that has already been disconnected, then merge the set of terminals on $v$ with the set of terminals on the parent of $v$.
  If $v$ is a terminal root for a terminal set that has not been cut yet, then add $v$ to the cut and mark all terminal sets that have a terminal on $v$ as cut.
  Observe that the instance must at least store (say, in a queue) on each node which terminal sets contain that node as a terminal (or, equivalently, for each terminal set all nodes that it contains).
  Therefore, the above algorithm indeed runs in linear time in terms of the input size $n + m + \sum_{T \in \mc{T}} |T|$.
\end{proof}

\subsection{Parameterized Analysis for \problemESC{} and \problemRNSC{}}
The situation for \problemESC{} and \problemRNSC{} is completely different than for \problemNSC{}.
C{\u a}linescu \etal\cite{CalinescuEtAl2003} proved that both \problemEMC{} and \problemRNMC{} are $\mathsf{NP}$-complete on trees.
This implies that \problemESC{} and \problemRNSC{} are para-$\mathsf{NP}$-complete on trees for parameter $p$.
In this subsection, we explore all possible other parameterizations for these problems on trees.

\subsubsection{Parameter \texorpdfstring{$k$}{k}}
We briefly recall the definition of the \problemHS{} problem, which is known to be $\mathsf{W}[2]$-complete~\cite{DowneyFellows1999}.
Recall that \problemHS{} is the problem of, given a universe $U$, a family $\mc{F}$ of subsets of $U$, and an integer $k$, to decide whether there is a set $H \subseteq U$ with $|H| \leq k$ such that $F \cap H \not=\emptyset$ for each $F \in \mc{F}$ (i.e., $H$ is a hitting set of size at most $k$).

\begin{theorem} \label{thm:edge:hard-k}  \label{thm:rnode:hard-k}
  \problemESC{} and \problemRNSC{} are $\mathsf{W}[2]$-hard on trees for the parameter $k$, even if the terminal sets are pairwise disjoint.
\end{theorem}
\begin{proof}
  We reduce from \problemHS{}.
  Let $(U, \mc{F},k)$ be an instance of this problem, and let $\mc{F}(u) = \{ F \in \mc{F} \mid u \in F \}$.
  We now build the following graph. For each $u \in U$, we build a path~$P_{u}$ on $|\mc{F}(u)|$ nodes.
  Further, let $\sigma_{u} : \mc{F}(u) \rightarrow V(P_{u})$ be an arbitrary bijection.
  We then add a root $r$, and for each $u \in U$, we connect $r$ to an end of $P_{u}$ through a new edge $e_{u}$.
  This is the tree $G$.
  We now build terminal sets.
  Let $T_{F} = \{r\} \cup \{ \sigma_{u}(F) \mid \forall u \in F \}$ for any $F \in \mc{F}$, and let $\mc{T} = \{T_{F} \mid F \in \mc{F} \}$.
  The final instance of \problemESC{} is $(G,\mc{T},k)$.

  Suppose that $\mc{F}$ has a hitting set $H$ of size at most $k$.
  Let $S = \{e_{u} \mid u \in H\}$.
  Note that if some terminal set $T_{F}$ is not disconnected in $G \setminus S$, then this contradicts that $H \cap F \not= \emptyset$.
  Further, $|S| \leq k$, and thus $S$ is a solution for $(G,\mc{T},k)$.

  Suppose that $(G,\mc{T},k)$ has a solution, and let $S$ be a solution such that $S \cap \{e_{u} \mid u \in V(G)\}$ is maximum.
  Since $e_{u}$ disconnects at least the terminal sets disconnected by any edge of $P_{u}$, $S \setminus \{e_{u} \mid u \in U\} = \emptyset$.
  Let $H = \{u \mid e_{u} \in S\}$.
  Note that if $H \cap F = \emptyset$ for some $F \in \mc{F}$, then $T_{F}$ is not disconnected by $S$.
  Further, $|H| \leq k$, and thus $H$ is solution for $(U,\mc{F},k)$.

  A modification of the reduction yields the result for \problemRNSC{}.
  Let $(U, \mc{F},k)$ be an instance of \problemHS{}.
  We again build the paths $P_{u}$ and the terminal sets $T_{F}$ (which include the root $r$).
  However, we do not connect $r$ to an end of $P_{u}$ directly; instead, for each $u \in U$, we add a node $n_{u}$, and connect $n_{u}$ to both $r$ and an end of $P_{u}$.
  Let~$G$ denote the resulting tree, and let $\mc{T} = \{T_{F} \mid F \in \mc{F} \}$.
  The final instance of \problemRNSC{} is $(G,\mc{T},k)$.

  Suppose that $\mc{F}$ has a hitting set $H$ of size at most $k$.
  Let $S = \{n_{u} \mid u \in H\}$.
  Note that if some terminal set~$T_{F}$ is not disconnected in $G \setminus S$, then this contradicts that $H \cap F \not= \emptyset$.
  Further, $|S| \leq k$, and thus $S$ is a solution for $(G,\mc{T},k)$.

  Suppose that $(G,\mc{T},k)$ has a solution $S$.
  Since the only non-terminal nodes are the $n_{u}$, we let $H = \{u \mid n_{u} \in S\}$ and note that $|H| = |S| \leq k$.
  Further, if $H \cap F = \emptyset$ for some $F \in \mc{F}$, then~$T_{F}$ is not disconnected by $S$.
  Hence, $H$ is solution for $(U,\mc{F},k)$.
\end{proof}
Note that in the constructions of both theorems, we can obtain a slightly different tree structure by replacing the paths $P_{u}$ with stars.
Further, if we do not insist that the terminal sets are pairwise disjoint, then both constructions can be simplified to a star by contracting each path~$P_{u}$ into a single node.

\subsubsection{Kernelizaton for \texorpdfstring{$t$}{t} and \texorpdfstring{$k+t$}{k+t}}
We now consider whether \problemESC{} and \problemRNSC{} have a polynomial kernel on trees for the parameter $t$.
We prove the following stronger result.

\begin{theorem}
\label{thm:edge:nokernel-kt}
\label{thm:rnode:nokernel-kt}
  \problemESC{} and \problemRNSC{} have no polynomial kernel on trees for the parameter $k+t$, even if the terminal sets are pairwise disjoint, unless the polynomial hierarchy collapses to the third level.
\end{theorem}
Note that this theorem indeed implies that \problemESC{} and \problemRNSC{} have no polynomial kernel on trees for the parameter $t$.

We need the observation that \problemHS{} has no polynomial kernel for the parameter solution size plus the number of sets, unless the polynomial hierarchy collapses to the third level.
To see this, Dom~\etal\cite[Theorem~2]{DomEtAl2009} prove that \problemSetC{} has no polynomial kernel for the parameter solution size plus the size of the universe, unless the polynomial hierarchy collapses to the third level.
Consider an instance $(E,\mc{S}, k)$ of \problemSetC{}.
For each $S \in \mc{S}$, create an element $u_{S}$, and let $U$ be the set of all these elements.
For each $e \in E$, create a set $F_{e} = \{ u_{S} \mid e \in S\}$, and let $\mc{F}$ be the family of all these sets.
Then, $(E,\mc{S},k)$ is a ``yes''-instance of \problemSetC{} if and only if $(U,\mc{F},k)$ is a ``yes''-instance of \problemHS{}.
Note that $|\mc{F}| = |E|$.
This completes a polynomial parameter transformation.
Since both \problemHS{} and \problemSetC{} are $\mathsf{NP}$-complete, the observation follows from Bodlaender~\etal\cite{BodlaenderEtAl2011}.

\begin{proof}[Proof of Theorem~\ref{thm:edge:nokernel-kt}]
  Note that the reductions of Theorem~\ref{thm:edge:hard-k} and Theorem~\ref{thm:rnode:hard-k} give polynomial parameter transformations for $k+t$.
  Further, \problemESC{} and \problemRNSC{} on trees and \problemHS{} are $\mathsf{NP}$-complete.
  The theorem then follows from Bodlaender et al.~\cite{BodlaenderEtAl2011}.
\end{proof}

\subsubsection{Parameter \texorpdfstring{$k+p$}{k+p}}
This result generalizes the known algorithms for \problemEMC{} and \problemRNSC{} on trees for the parameter~$k$~\cite{GuoNiedermeier2005,GuoEtAl2006}.

\begin{theorem}
\label{thm:edge:easy-kp}
\label{thm:rnode:easy-kp}
  \problemESC{} and \problemRNSC{} are fixed-parameter tractable on trees for parameter $k+p$.
\end{theorem}
\begin{proof}
  Let $(G, \mc{T}, k)$ be an instance of \problemESC{}, where $G$ is a tree. Let $r$ and the terminal roots of $G$ be defined as before.
  We use the following crucial observation, which is similar to the crucial observation made in Theorem~\ref{thm:node:easy}: for any $T \in \mc{T}$, if the subtree $G' \supseteq G_{T}$ of $G$ rooted at $r(G_{T})$ contains no terminal roots except $r(G_{T})$, then there is an optimal solution that contains an edge of $G_{T}$ incident to $r(G_{T})$.
  To see this, note that any solution must contain at least one edge of $G_{T}$ in order to disconnect $T$.
  Further, since $G'$ has no terminal roots except $r(G_{T})$, the set of terminal sets disconnected when removing an edge $e \in E(G_{T})$ is a subset of the set of terminal sets disconnected when removing the edge incident to $r(G_{T})$ on the path in~$G$ from $e$ to $r(G_{T})$.

  Using the above crucial observation, a branching strategy becomes apparent.
  First, if $k=0$ and $\mc{T} \not=\emptyset$, then return ``no''; if $k\geq 0$ and $\mc{T} = \emptyset$, then return ``yes''. Otherwise, compute the terminal roots of each terminal set.
  Then, find a terminal set $T$ for which $r(G_{T})$ is deepest in $G$.
  Since $|T| \leq p$, $G_{T}$ has at most $p$ edges incident on $r(G_{T})$.
  Branch on all such edges $e$.
  Then, add~$e$ to the cut, and recurse on the instance $(G', \mc{T}', k-1)$, where $G'$ is the connected component of $G \setminus e$ that contains $r(G_{T})$, and $\mc{T}'$ is $\mc{T}$ minus every terminal set that is disconnected by $e$.
  If any of the branches results in a ``yes'', then return ``yes''; otherwise, return ``no''.

  Correctness of the algorithm is immediate from the crucial observation.
  Further, observe that the algorithm has at most $p^{k}$ branches.

  To prove the theorem for \problemRNSC{}, we essentially combine the crucial observation of Theorem~\ref{thm:node:easy} with the above crucial observation for the edge version.
  Let $(G, \mc{T}, k)$ be an instance of \problemRNSC{}, where $G$ is a tree.
  As a preprocessing step, we contract any edge of which both endpoints are a terminal.
  Since we are not allowed to delete terminals, this is safe.
  However, if such a contraction makes a terminal set become a singleton set, then we may return ``no''.
  By abuse of notation, we call the resulting tree $G$ and the resulting family of terminal sets $\mc{T}$.
  Root $G$ at a node $r$, and define $G_{T}$ and terminal roots $r(G_{T})$ as before.

  We follow a branching strategy.
  First, if $k=0$ and $\mc{T} \not=\emptyset$, then return ``no''; if $k\geq 0$ and $\mc{T} = \emptyset$, then return ``yes''.
  Otherwise, compute the terminal roots of each terminal set.
  Then, find a terminal set $T$ for which $r(G_{T})$ is deepest in $G$.
  If $r(G_{T})$ is not a terminal node, then following the crucial observation of Theorem~\ref{thm:node:easy}, we add $r(G_{T})$ to the cut, and recurse on the instance $(G', \mc{T}', k-1)$, where $G'$ is $G$ minus the subtree rooted at $r(G_{T})$, and $\mc{T}'$ is $\mc{T}$ minus every terminal set that is disconnected by $r(G_{T})$.
  If $r(G_{T})$ is a terminal node, then none of the children of $r(G_{T})$ are a terminal.
  Since $|T| \leq p$, $r(G_{T})$ has at most $p$ neighbors in $G_{T}$.
  Branch on all such neighbors $c$.
  Then, add $c$ to the cut, and recurse on the instance $(G', \mc{T}', k-1)$, where~$G'$ is $G$ minus the subtree rooted at $c$, and $\mc{T}'$ is $\mc{T}$ minus every terminal set that is disconnected by $c$.
  If any of the branches results in a ``yes'', then return ``yes''; otherwise, return ``no''.

  Correctness of the algorithm is immediate from the crucial observation of Theorem~\ref{thm:node:easy} and the above crucial observation for the edge version.
  Further, observe that the algorithm has at most $p^{k}$ branches.
  The theorem follows.
\end{proof}

Since \problemESC{} and \problemRNSC{} are both fixed-parameter tractable on trees for the parameter $k+p$, it is natural to ask whether these problems admit a polynomial kernel.
We answer this question negatively.

\begin{theorem}
\label{thm:edge:nokernel-kp}
\label{thm:rnode:nokernel-kp}
  \problemESC{} and \problemRNSC{} do not admit polynomial kernels on trees for parameter $k+p$, even if the terminal sets are pairwise disjoint, unless the polynomial hierarchy collapses to the third level.
\end{theorem}
\begin{proof}
  It is known that \problemHS{} has no polynomial kernel when parameterized by the solution size and the maximum size of any set, unless the polynomial hierarchy collapses to the third level~\cite[Theorem~6]{DomEtAl2009}.
  Moreover, \problemESC{} and \problemRNSC{} on trees and \problemHS{} are $\mathsf{NP}$-complete.
  Since the reductions given in Theorem~\ref{thm:edge:hard-k} and Theorem~\ref{thm:rnode:hard-k} are polynomial parameter transformations from \problemHS{} for the parameter solution size plus maximum set size to \problemESC{} and \problemRNSC{}, respectively, for parameter $k+p$, the theorem follows from Bodlaender~\etal\cite{BodlaenderEtAl2011}.
\end{proof}

\subsubsection{Parameter \texorpdfstring{$t$}{t}}
We note that the result of the following theorem is actually dominated by Theorem~\ref{thm:t+tw}.
However, since the theorem is much simpler to prove for trees and does not rely on MSOL, we give its proof for completeness.

\begin{theorem}
\label{thm:edge:easy-t}
\label{thm:rnode:easy-t}
  \problemESC{} and \problemRNSC{} are fixed-parameter tractable on trees for parameter $t$.
\end{theorem}
\begin{proof}
  Let $(G,\mc{T},k)$ be an instance of \problemESC{}, where $G$ is a tree.
  For each $T \in \mc{T}$, let $G_{T}$ again denote the subtree of $G$ induced by the terminals in $T$.
  Let $\mc{F} = \{ E(G_{T}) \mid T \in \mc{T}\}$.
  Observe that $(E(G), \mc{F}, k)$ is a ``yes''-instance of \problemHS{} if and only if $(G,\mc{T},k)$ is a ``yes''-instance of \problemESC{}.
  Recall that \problemHS{} is fixed-parameter tractable parameterized by the number of sets~\cite{FominEtAl2004}.
  Since $|\mc{F}| = t$, the theorem follows.

  This algorithm can be easily adapted for \problemRNSC{}: the universe of the \problemHS{} instance becomes $V(G)$, and the family of sets becomes $\{ V(G_{T}) - X \mid T \in \mc{T}\}$, where $X$ is the set of terminal nodes.
\end{proof}

To contrast Theorem~\ref{thm:edge:nokernel-kt}, we prove that both \problemESC{} and \problemRNSC{} admit a polynomial kernel on trees for parameter $t + p$.
\begin{theorem} \label{thm:edge:kernel-tp}
  \problemESC{} admits a polynomial kernel on trees for parameter $t+p$.
\end{theorem}
\begin{proof}
  Let $(G, \mc{T}, k)$ be an instance of \problemESC{}, where $G$ is a tree.
  We describe two reduction rules.

  The first reduction rule concerns leafs.
  Let $v$ be a leaf of $G$ and suppose that $v$ is not a terminal.
  Let $u$ be the neighbor of $v$ in $G$.
  Then no minimal solution contains the edge $(u,v)$.
  Hence, we may remove $v$ without changing the feasibility of the instance.
  Apply this reduction rule exhaustively, until each leaf of $G$ is a terminal.

  The second reduction rule concerns nodes of degree two.
  Consider any triple of nodes $u,v,w \in V(G)$ such that $N(v) = \{u,w\}$ and $v$ is not a terminal.
  Then no minimal solution selects both $(u,v)$ and $(v,w)$.
  Hence, we may remove $v$ and connect $u,w$ without changing the feasibility of the instance.
  Apply this reduction rule exhaustively, until each internal node of $G$ is either a terminal or of degree at least three.

  Observe that after applying the first reduction rule exhaustively, $G$ has $O(tp)$ leafs.
  After applying the second reduction rule, $G$ has $O(tp)$ internal nodes that are terminals.
  Further, any non-terminal internal node has degree at least three.
  Since the number of non-terminal internal nodes cannot exceed the number of leafs, $G$ has $O(tp)$ nodes in total.
\end{proof}

\begin{theorem}
\label{thm:rnode:kernel-tp}
  \problemRNSC{} has a polynomial kernel on trees for parameter $t+p$.
\end{theorem}
\begin{proof}
  Let $(G, \mc{T}, k)$ be an instance of \problemRNSC{}, where $G$ is a tree.
  We describe two reduction rules. 
  The first reduction rule is the same as in Theorem~\ref{thm:edge:kernel-tp}: remove any leaf $u$ that is not a terminal.
  This rule is safe, as no minimal solution contains $u$.
  Apply this reduction rule exhaustively, until each leaf of $G$ is a terminal.

  The second reduction rule is a modified version of the second reduction rule of Theorem~\ref{thm:edge:kernel-tp}.
  Consider any triple of distinct nodes $(u,v,w) \in V(G)^3$ such that $N(v) = \{u,w\}$ and $u$ and $v$ are not a terminal.
  Then no minimal solution selects both $u$ and $v$.
  Hence, we may remove $v$ and connect $u,w$ without changing the feasibility of the instance.
  Apply this reduction rule exhaustively, until each internal node of $G$ either is a terminal, or is of degree at least three, or is of degree two and has only terminal neighbors.

  Observe that after applying the first reduction rule exhaustively, $G$ has $O(tp)$ leafs.
  After applying the second reduction rule, $G$ has $O(tp)$ internal nodes that are terminals.
  Further, any non-terminal internal node of degree two has a terminal as child, and thus their number is bounded by $O(tp)$ as well.
  Finally, the number of non-terminal internal nodes of degree at least three cannot exceed the number of leafs.
  Hence, $G$ has $O(tp)$ nodes in total.
\end{proof}

\section{Discussion}
\label{sec:discussion}

We provided a comprehensive computational complexity analysis of the {\sc Steiner Multicut} problem with respect to fundamental parameters, culminating in either a fixed-parameter algorithm or a $\mathsf{W}[1]$-hardness result for \emph{every} combination of parameters.
This way, we generalize known tractability results for special cases of {\sc Steiner Multicut}, and chart the boundary of tractability for other cases. See Table~\ref{tab:undirected_parameterized_complexity_results} for a complete overview.

We leave several interesting questions for future research. A first possible extension of our work is to consider generalizations of \problemSC{}, for example to \textsc{Requirement Multicut}, where each terminal set $T_i$ must be cut into $r_{i}\geq 2$ components.
The approximability of this problem has been considered in several papers~\cite{NagarayanRavi2010,Moitra2009,GuptaEtAl2010}.
From a parameterized viewpoint, we note that all hardness results of this paper carry over to \textsc{Requirement Multicut}.
Yet, it is an intriguing open question which of our fixed-parameter algorithms generalize to {\sc Requirement Cut}.

A second possible extension is to consider directed graphs.
Already \problemMC{} is $\mathsf{W}[1]$-hard in this case~\cite{MarxRazgon2011} for parameter cut size $k$, even on acyclic directed graphs~\cite{KratschEtAl2012}.
On the other hand, \problemMC{} is fixed-parameter tractable for the parameter $k + t$ in directed acyclic graphs~\cite{KratschEtAl2012}.
It would be interesting whether this result generalizes to \problemSC{}.

\begin{table}[ht]
  \centering
  \begin{tabular}{ll|c}
    \toprule
      constants            & params & poly kernel?\\
      \midrule
      $\tw(G)=1$           & $t,p$   & $\checkmark$ (Thm.~\ref{thm:edge:kernel-tp},~\ref{thm:rnode:kernel-tp}) \\
      
      $\tw(G)=1$           & $k,p$   & $\mathsf{X}$ (Thm.~\ref{thm:edge:nokernel-kp}) \\
      
      $\tw(G)=1$           & $k,t$   & $\mathsf{X}$ (Thm.~\ref{thm:edge:nokernel-kt}) \\
      
      $\tw(G)=1,p \geq 3$  & $k$     & open \\
      
      $\tw(G)=1,p=2$       & $k$     & $\checkmark$~\cite{BousquetEtAl2009,ChenEtAl2012} \\
      
      $p=2, t$             & $k$     & $\checkmark$~\cite{KratschWahlstrom2012} \\
      
      $p=2$                & $k$     & $\mathsf{X}$~\cite{CyganEtAl2012} \\
      
                           & $k,t,p$ & open \\
      
      $p=2$                & $k,t$   & open \\
      
      $p=3, t=3$           & $k$     & open \\
      
      $t$                  & $k$     & open \\
    \bottomrule
  \end{tabular}
  \centering
  \caption{A summary of results about the kernelization of \problemESC{}; the top three lines also hold for \problemRNSC{}. We note that the kernelization complexity of \problemESC{} and \problemRNSC{} are completely characterized for trees, except that the case that $p \geq 3$ is a constant and $k$ is a parameter is still open. We have listed all minimal and maximal open questions (among all cases with $\tw(G) = 1$ or unbounded).}
\label{tab:kernelization}
\end{table}

A third possible extension is to investigate which problems admit polynomial kernels. 
While we have resolved many kernelization questions in this paper, several open problems remain (see Table~\ref{tab:kernelization}), in particular whether there is a polynomial kernel for the parameters $k+t+p$ on general graphs.
Answers in this research direction might shed new light on some long-standing open questions~\cite{CyganEtAl2013b} on the existence of polynomial kernels for \problemMC{} for parameter $k + t$ (currently, only a kernel of size $k^{O(\sqrt{t})}$ is known~\cite{KratschWahlstrom2012}, and there is no kernel of size polynomial in $k$ only~\cite{CyganEtAl2012}).

\medskip
\noindent
\textbf{Acknowledgements.}
We thank Magnus Wahlstr{\"o}m for an insight that helped in proving Lemma~\ref{lem:NAEhard}.

\bibliographystyle{abbrv}
\bibliography{steinercuts}

\appendix

\section{Basic Definitions}
\label{sec:basic}

In this section, we give formal definitions of several core notions that are used in this paper.
First and foremost, a \emph{parameter} of a problem instance $I$ of problem $\Pi$ is an integer $k(I)$.
Usually, this parameter is part of the problem definition, so we commonly write it as just $k$, and it is often equal to the size of the optimum or desired solution to the problem instance.
A problem is \emph{fixed-parameter tractable} if it admits a \emph{fixed-parameter algorithm}, which decides instances $I$ of $\Pi$ in time $f(k(I)) \cdot |I|^{O(1)}$-time for some computable function $f$.
The class of fixed-parameter tractable problems is denoted by $\mathsf{FPT}$.
As evidence that a problem $\Pi$ is unlikely to be fixed-parameter tractable, one can show that $\Pi$ is $\mathsf{W}[1]$- or $\mathsf{W}[2]$-hard. Informally, $\mathsf{W}[1]$ is the class of problems that are as hard as \textsc{Clique} parameterized by the maximum clique size, and $\mathsf{W}[2]$ is the class of problems that are as hard as \problemHS{} parameterized by the minimum hitting set size.
To prove hardness of $\Pi$, one can give an \emph{parameterized reduction} from a $\mathsf{W}[\cdot]$-hard problem~$\Pi'$ to $\Pi$ that maps every instance $I'$ of $\Pi'$ with parameter $k'$ to an instance $I$ of $\Pi$ with parameter $k \leq g(k')$ for some computable function $g$ such that $I'$ can be computed in time $f(k') \cdot |I|^{O(1)}$ for some computable function $f$, and $I$ is a ``yes''-instance if and only if $I'$ is a ``yes''-instance.
In case~$f$ and $g$ are polynomials, such a reduction is called a \emph{polynomial parameter transformation}.
A second way to prove hardness is by showing that the problem is $\mathsf{NP}$-complete even if the parameter $k$ is a constant; in that case, the parameterized problem is called \emph{para-$\mathsf{NP}$-complete}.

An equivalent notion to fixed-parameter tractability is that of a ``kernel''.
Given an instance~$I$ with parameter $k$ of a problem $\Pi$, a \emph{kernelization algorithm} computes in polynomial time an instance $I'$ (a \emph{kernel}) of $\Pi$ with size $|I'| \leq f(k)$, for some computable function $f$, such that $I$ is a ``yes''-instance if and only if $I'$ is.
Any decidable fixed-parameter tractable problem for a parameter $k$ with running time $f(k) \cdot n^{O(1)}$ admits a kernel with the same function $f$. Conversely, no $\mathsf{W}[\cdot]$-hard problem admits a kernel, unless $\mathsf{FPT} = W[\cdot]$.
For $\mathsf{NP}$-hard problems, a fixed-parameter algorithm can generally only yield kernels with an superpolynomial function~$f$, unless $\mathsf{P} = \mathsf{NP}$.
However, using different techniques it is sometimes possible to obtain a \emph{polynomial kernel}, for which the function $f$ is a polynomial.
For yet other problems, the existence of a polynomial kernel is unlikely~\cite{Drucker2012,BodlaenderEtAl2009,FortnowSanthanam2011}, for example because the polynomial hierarchy would collapse to the third level.
In this paper, we use that for a parameterized problem $\Pi'$, for which a polynomial kernel is known to be unlikely, and a polynomial parameter transformation from~$\Pi'$ to a parameterized problem $\Pi$, it is also unlikely that $\Pi$ has a polynomial kernel~\cite{BodlaenderEtAl2011}.

A good introduction into parameterized complexity can be found in the books by Downey and Fellows~\cite{DowneyFellows1999} and Flum and Grohe~\cite{FlumGrohe2006}.

The final notion that we define is that of ``treewidth''.
Given an undirected graph $G$, a \emph{tree decomposition} of $G$ is a tree $T$ and a collection of \emph{bags} $\mc{B} = \{B_{t} \subseteq V(G) \mid t \in V(T) \}$ such that
\begin{itemize}[noitemsep]
  \item for each node $v \in V(G)$, the set $\{t \mid v \in B_{t}\}$ is a nonempty set that induces a connected subtree of $T$;
  \item for each edge $\{u,v\} \in E(G)$, there is a bag $B_{t}$ such that $u,v \in B_{t}$.
\end{itemize}
The \emph{width} of a tree decomposition $(T,\mc{B})$ is $\max_{B \in \mc{B}} \{|B| - 1\}$, and the \emph{treewidth} $\tw(G)$ of $G$ is the minimum width of all tree decompositions of $G$.
We refer to Bodlaender~\cite{Bodlaender1998} for some basic properties of tree decompositions.

\knip{
\begin{table}[t]
  \centering
  \begin{tabular}{|ll|c|c|c|}
    \toprule
                           &            & {\sc Edge}                                 & {\sc Node}       & {\sc Restr.\ Node}\\
      constants            & params & {\sc Steiner MC}                     & {\sc Steiner MC}        & {\sc Steiner MC}\\
      
    \midrule
       
      $p=2$ & $k$ & $\mathsf{X}$~\cite{CyganEtAl2012} & $\mathsf{X}$~\cite{CyganEtAl2012} & $\mathsf{X}$~\cite{CyganEtAl2012} \\

    \bottomrule
  \end{tabular}
  
  \centering
  \caption{A summary of results about the kernelization of \problemESC{}; the top four lines also hold for \problemRNSC{} and the bottom three lines hold for both node deletion versions of \problemNSC{}. We note that the kernelization complexity of \problemESC{} and \problemRNSC{} are completely characterized for trees, except that the case that $p \geq 3$ is a constant and $k$ is a parameter is still open. We have not listed the several open questions for the parameter treewidth or treewidth greater than one.}
\label{tab:kernels}
\end{table}
}

\section{Tree Diagrams}
We provide tree diagrams of Table~\ref{tab:undirected_parameterized_complexity_results} for each of the three versions of \problemSC{}. The diagrams exhibit both the $\mathsf{FPT}$ versus $\mathsf{W[1]}$-hardness/para-$\mathsf{NP}$-completeness dichotomy and the polynomial-time versus $\mathsf{NP}$-completeness dichotomy. To see the latter, remove all branches marked only `param' from each of the diagrams.

In the diagrams, ``T.x'' refers to Theorem~x, ``S.x'' refers to Section~x, ``const'' means that the preceding variable is taken to be a fixed constant, ``param'' means that the preceding variable is taken to be a parameter, and ``unb'' (unbounded) means that we make no assumptions on the preceding variable.

\begin{figure}[hp]
\begin{center}
\ig[scale=1.25]{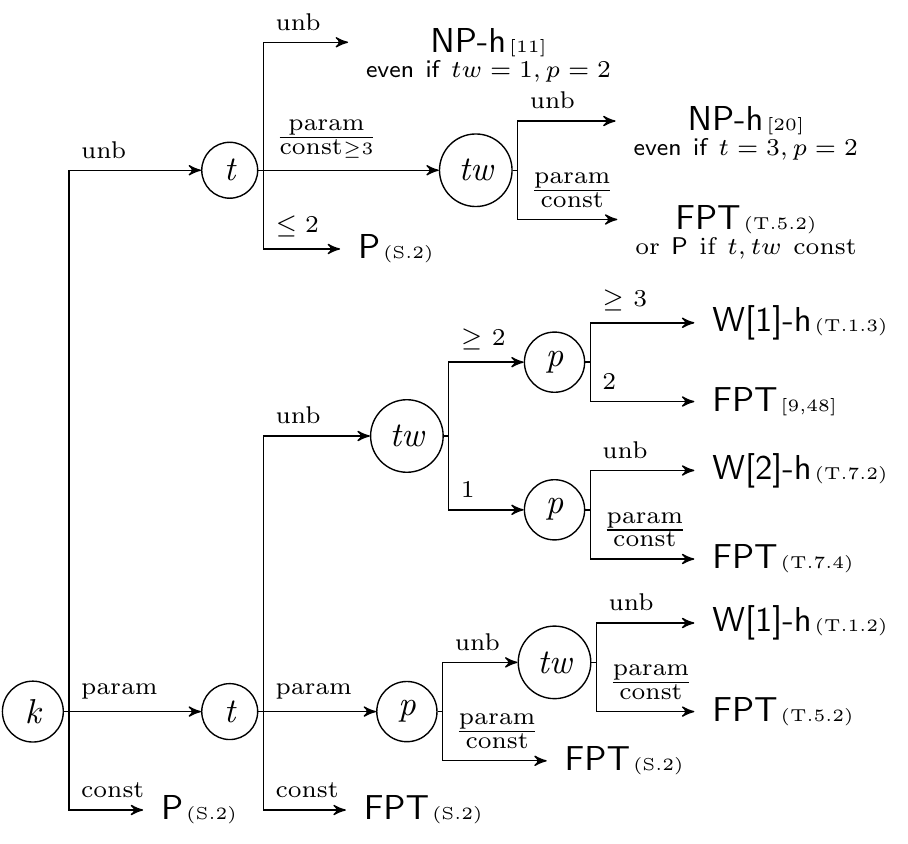}
\end{center}
\caption{Tree diagram for \problemRNSC{}.}
\end{figure}

\label{sec:trees}
\begin{figure}[hp]
\begin{center}
\ig[scale=1.25]{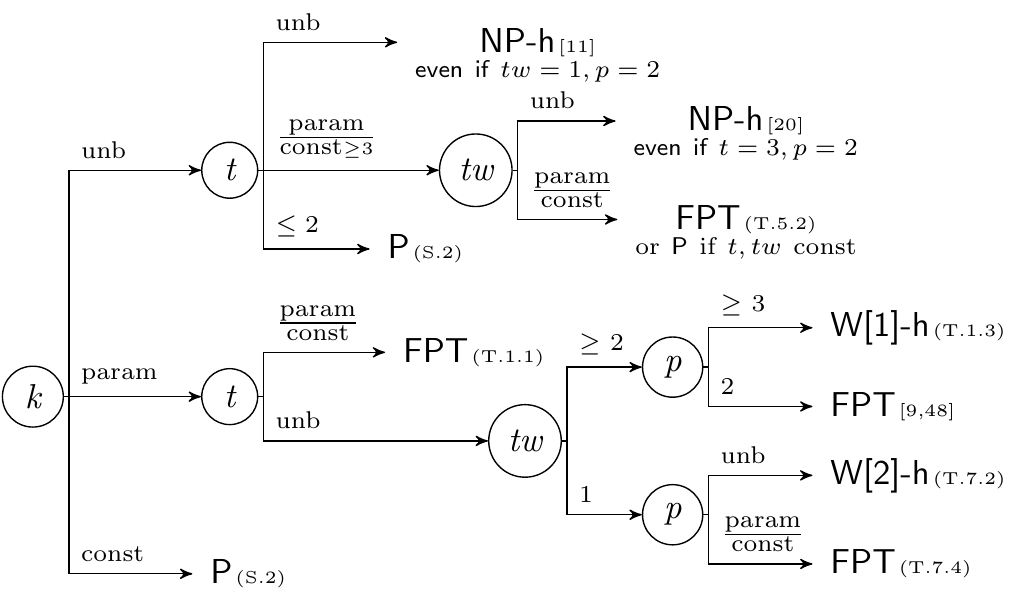}
\end{center}
\caption{Tree diagram for \problemESC{}.}
\end{figure}

\begin{figure}[hp]
\begin{center}
\ig[scale=1.25]{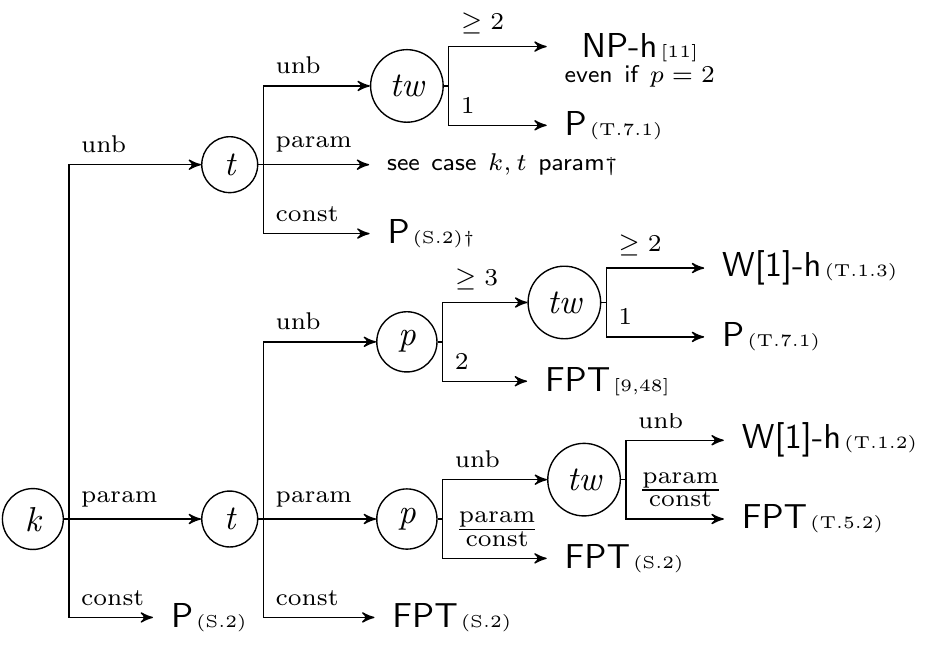}
\end{center}
\caption{Tree diagram for \problemNSC{}. For the cases marked with a $\dagger$, one has to realize that an instance with $k \geq t$ is always a ``yes''-instance, and thus we may always assume that $k < t$.}
\end{figure}

\end{document}